\documentclass[11pt,onecolumn]{article}

\usepackage{amsmath}
\usepackage{amssymb}
\usepackage{amsthm}

\usepackage{microtype}
\usepackage{times}

\usepackage{algorithmicx}
\usepackage{algorithm}
\usepackage[noend]{algpseudocode}

\usepackage{multirow}
\usepackage{xspace}
\usepackage{graphicx}
\usepackage{ifpdf}
\usepackage{latexsym}
\usepackage{euscript}
\usepackage{xspace}
\usepackage{color}
\usepackage{makeidx}
\usepackage{picins,wrapfig}
\usepackage{xypic}
\usepackage{mathtools}
\PassOptionsToPackage{hyphens}{url}\usepackage{hyperref}

\usepackage{subcaption}

\graphicspath{ {figures/} }

\theoremstyle{plain}
\newtheorem{theorem}{Theorem}[section]

\newtheorem{lemma}[theorem]{Lemma}

\newcommand{\denselist}{\itemsep 0pt\parsep=1pt\partopsep 0pt}
\newcommand{\etal}              {et al. \xspace}

\newcommand{\eps}{{\varepsilon}}
\newcommand{\Rips}	{\mathcal{R}}
\newcommand{\SRips}	{\mathcal{S}}
\newcommand{\QRips}	{\mathcal{Q}}
\newcommand{\SetRips}{\mathcal{B}}
\newcommand{\hh}		{{\hat{h}}}
\newcommand{\vmap}	{{\pi}}
\newcommand{\hv}		{{\hat{\vmap}}}
\newcommand{\hmap}	{{h}}
\newcommand{\oldhmap}		{{s}}
\newcommand{\dpert}{\hat{d}}

\DeclareMathOperator*{\argmax}{argmax}
\DeclareMathOperator*{\argmin}{argmin}

\begin{document}
	
\title{\textbf{SimBa}: An Efficient Tool for Approximating Rips-filtration Persistence via \underline{\textbf{Sim}}plicial \underline{\textbf{Ba}}tch-collapse}
\author{Tamal K. Dey\thanks{Department of Computer Science and Engineering, The Ohio State University, Columbus, OH, USA. Emails: \texttt{tamaldey, shiday, yusu@cse.ohio-state.edu}},~~~ Dayu Shi,~~ Yusu Wang$^*$}

\date{}

\maketitle

\begin{abstract}
In topological data analysis, a point cloud data $P$ extracted from a metric space is often analyzed by computing the persistence
diagram or barcodes of a sequence of Rips complexes built on $P$ indexed by a scale parameter.  
Unfortunately, even for input of moderate size, the size of the
Rips complex may become prohibitively large as the scale parameter increases. 
Starting with the \emph{Sparse Rips filtration} introduced by Sheehy, some existing methods aim to reduce the size of the complex so as to improve the time efficiency as well. 
However, as we demonstrate, existing approaches still fall short of
scaling well, especially for high dimensional data. 
In this paper, we investigate the advantages and limitations of existing approaches. Based on insights gained from the experiments, we propose an efficient new algorithm, called \emph{SimBa}, for approximating the persistent homology of Rips filtrations with quality guarantees. Our new algorithm leverages a batch collapse strategy as well as a new sparse Rips-like filtration. We experiment on a variety of low and high dimensional data sets. We show that our strategy presents a significant size reduction, and our algorithm for approximating Rips filtration persistence is order of magnitude faster than existing methods in practice. 
 \end{abstract}

\section{Introduction}

In recent years, topological ideas and methods have emerged as a new paradigm for analyzing complex data \cite{Carl09,EH10}. An important line of work in this direction is the theory and applications of persistent homology. It provides a powerful and flexible framework to inspect data for characterizing and summarizing important features that persist across different scales. 
Since its introduction \cite{ELZ02,F90,Robins99}, there have been many fundamental developments \cite{BD11,zigzag,CZ09,CSGGO2009,CSGO12,CEH07,CEH09,CEM06,ZC05} both to generalize the framework and to provide theoretical understanding for various aspects of it (such as its stability). 
These developments help to provide foundation and justification of the practical usage of persistent homology; see e.g, \cite{dgh-rips,CGOS11,DLLRSW10,dario,RHBK15}. 

A determining factor in applying persistent homology to a broad range of data analysis problems is the availability of efficient and scalable software.
Given the rapidly increasing size of modern data, the "efficiency" necessarily concerns both time and space complexities. 
The original algorithm to compute persistent homology takes $O(n^3)$ time and $O(n^2)$ space for a filtration involving $n$ number of total simplices \cite{ELZ02}. 
Various practical improvements have been suggested \cite{CK13,DMV11}. 
An early software widely used for computing persistence is
Morozov's Dionysus ~\cite{Dionysus}.
Later, Bauer \etal{} developed the PHAT toolbox \cite{BKR14b},  
based on several efficient matrix reduction strategies (mostly focusing on time efficiency) as described in \cite{BKR14a}. 
A more recently developed library called GUDHI \cite{GUDHI2015} considers the improvement both in time and space efficiencies. In particular, it uses an efficient data structure, called the simplex tree \cite{BM2012}, to encode input simplicial complexes, and uses the \emph{compressed annotation matrix} \cite{BDM2013} to implement the persistent cohomology algorithm. 
Dionysus, PHAT, and GUDHI offer efficient software for computing
persistence induced by inclusions. For our algorithm, we need
persistence induced by more general simplicial maps for which
we use Simpers~\cite{Simpers2015} developed on the basis of the 
algorithm in~\cite{DFW2014}.

The above results and software cater to general persistence computations. In practice, often the persistence needs to be computed for
a particular filtration called the \emph{Vietoris-Rips} or 
\emph{Rips filtration} in short. Given a set of points $P$ embedded in $\mathbb{R}^d$ (or in more general metric spaces), the Rips complex $\Rips^{\alpha}(P)$ with radius or scale $\alpha$ is the clique complex induced by the set of edges $\{ (p, p') \mid d(p, p') \le \alpha, p, p' \in P\}$. 
One is interested in the persistent homology induced by the sequence of Rips complexes $\Rips^{\alpha_1} \subseteq \Rips^{\alpha_2} \subseteq \cdots \subseteq \Rips^{\alpha_m}$ for a growing sequence of radii $\alpha_1 \le \alpha_2 \le \ldots \le \alpha_m$. 
Intuitively, the Rips complex at a specific scale $\alpha$ approximates the union of radius-$\alpha$ balls around sample points in $P$.
Thus, it captures the structure formed by input points $P$ at different scales. 

Unfortunately, even for a modest size of $n$ (in the range of thousands), the size of Rips complex (as well as the slightly more economical \v Cech complex) becomes prohibitively large as the radius $\alpha$ increases. 
In \cite{S2012}, Sheehy proposed an elegant solution for this problem by introducing a \emph{sparse Rips filtration} to approximate the persistent homology of the Rips filtration for a set of points $P$. 
An alternative approach of collapsing input points in batches with
increasing radius $\alpha$ was proposed in \cite{DFW2014}, 
which leveraged the persistence algorithm proposed
in the same paper for filtrations arising out of simplicial maps. 

\paragraph*{New work.} 
Given the importance of the Rips filtration in practice, our goal is to investigate the practical performance of the existing proposed methods, understand their advantages and limitations, and develop an efficient implementation for 
approximating the persistent homology of Rips filtrations. 
To this end, we make the following contributions. 
\begin{itemize}\denselist
	\item[1.] We investigate the advantages and limitations of 
	three existing methods, two based on Sparse Rips~\cite{S2012,CJS2015},
	and another on Batch-collapse~\cite{DFW2014}. Specifically, experiments show that while the sparse Rips algorithm by Sheehy \cite{S2012} has 
a theoretical guarantee on the size of the filtration and gives good approximation of the persistence diagrams for the Rips filtration in practice, it generates simplicial complexes of large size even for input of moderate size. This problem becomes more severe as the dimension of the input data increases. The algorithm fails to finish for several high dimensional data sets of rather moderate size. See Table \ref{fig:table} for some examples. The batch-collapse approach is much more space efficient (which leads to time efficiency as well). Nevertheless, we find that its size still becomes prohibitive for high dimensional data. 
	\item[2.] Based on the insights gained from experimenting with the existing approaches, we propose a new algorithm called \emph{SimBa} that approximates a Rips filtration persistence via simplicial batch-collapses. Our algorithm is a modification of the previous batch-collapse of Rips filtration proposed in \cite{DFW2014}. While theoretically, the modification may not seem major, empirically, it reduces the size of the filtration significantly and thus leads to a much more efficient approximation of the Rips filtration persistence. Furthermore, we show that this modification maintains a similar approximation guarantee as the batch-collapse of Rips filtration proposed in \cite{DFW2014}. 
	We describe the details of an efficient and practical implementation of SimBa, the software for which has been made publicly available from \cite{Simpers2015}. 
\end{itemize}

Two concepts, homology groups of a simplicial complex, and simplicial
maps between two complexes are used throughout this paper.
We refer the reader to any standard text such as~\cite{Munkres}
for details. We denote the $p$-dimensional homology group of a
simplicial complex ${\mathcal K}$ under $\mathbb{Z}_2$ coefficients by
$H_p({\mathcal K})$.

\section{Rips filtration and its approximation}
\label{sec:existing}

Given a set of points $P \subset \mathbb{R}^d$, let $\langle p_0, \ldots, p_s \rangle$ denote the $s$-dimensional simplex spanned by vertices $p_0, \ldots, p_s \in P$. 
The Rips complex at scalar $\alpha$ is defined as $\Rips^\alpha(P) = \{ \langle p_{0},\ldots, p_{s} \rangle \mid \|p_{i} - p_{j}\| \le \alpha, $ for any $i, j \in [0, s]\}$.   
Now consider the following {\em Rips filtration}: 
\begin{align}\label{eqn:Ripsfiltration}
\{\Rips^\alpha(P)\}_{\alpha > 0} ~&:=~ \Rips^{\alpha_1}(P) \hookrightarrow \Rips^{\alpha_2}(P) \cdots \hookrightarrow \Rips^{\alpha_n}(P) \cdots 
\end{align}
The inclusion maps between consecutive complexes above
induce homomorphisms between respective homology groups, 
giving rise to a so called {\em persistence module} for dimension $p$: 
\begin{align}
	H_p(\Rips^{\alpha_1}(P))\rightarrow H_p(\Rips^{\alpha_2}(P))\rightarrow\ldots
	\rightarrow H_p(\Rips^{\alpha_n}(P)) \cdots
\end{align}
If a homology class is created at $\Rips^{\alpha_i}(P)$ (i.e, does not have pre-image under homomorphism $H_p(\Rips^{\alpha_{i-1}}) \rightarrow H_p (\Rips^{\alpha_i})$) and dies entering $\Rips^{\alpha_j}(P)$ (i.e, its image vanishes under homomorphism $H_p(\Rips^{\alpha_{j-1}}(P)) \to H_p(\Rips^{\alpha_j}(P))$, then $\alpha_i$ is its \emph{birth time}, $\alpha_j$ is its \emph{death time}, and the difference $\alpha_j - \alpha_i$ is called the {\em persistence} of the class. In each dimension, the persistence barcodes capture the persistence of such homology classes by using a horizontal bar with left and right end points
at $\alpha_i$ and $\alpha_j$ respectively. 
These \emph{persistence barcodes} of the above Rips filtration are often the target summary of $P$ and/or of the space $P$ samples, which one wishes to compute in topological data analysis. 

The main bottleneck for computing the barcodes
of a Rips filtration stems from its size blowup. As the parameter
$\alpha$ grows, the Rips complex $\Rips^{\alpha}(P)$ can become
huge very quickly. To address this blowup in size, Sheehy~\cite{S2012} 
suggested a novel approach of sparsifying the point set $P$ as one proceeds
with increasing $\alpha$ in a way that does not alter the
barcodes too much. The idea is to replace the original Rips filtration
$\{\Rips^{\alpha}(P)\}_{\alpha>0}$ on $P$ with a sequence of smaller complexes 
$\{\SRips^{\alpha}\}_{\alpha>0}$ and show that the two sequences
{\em interleave} at the homology level. Then, appealing to the 
results of interleaving persistence modules~\cite{CSGGO2009}, one can show that
the barcodes of $\{\SRips^{\alpha}\}_{\alpha>0}$ approximate
those of $\{\Rips^{\alpha}(P)\}_{\alpha>0}$ reasonably.
The complexes $\SRips^{\alpha}$ are constructed as the {\em union} of Rips-like
complexes built on a sequence of subsets of $P$ 
rather than the entire set $P$. 

The union allows the complexes in 
$\{\SRips^{\alpha}\}_{\alpha>0}$ to be connected with inclusions
and hence permits using
efficient algorithms and software designed for inclusion induced persistence.
However, the size of 
$\SRips^{\alpha}$ may still be large due to the union operation. 
An alternative is to avoid the union operation but allow deletion or collapse of vertices (and simplices) at larger scale $\alpha$ \cite{CJS2015,S2012}
resulting into a sequence of Rips-like complexes connected with simplicial maps instead of inclusions. 
This approach, which we refer to as {\em Sparse Rips with collapse}, however achieves only moderate improvements in size reduction. 
We find that much more aggressive size reduction can be achieved 
by considering the
collapse in a batched fashion that gives rise to the approach of {\em Batch-collapsed Rips} \cite{DFW2014}. 

Finally, building on the batch-collapse idea, we propose a new approach,
called {\em SimBa} that significantly reduces the size of Rips-like complexes and their computations. 
This is achieved primarily by 
replacing inter-point distances with {\em set distances} while
computing the complexes. 
We prove that this approach
still provably approximates the barcodes of the original Rips filtration in sequence (\ref{eqn:Ripsfiltration}). 

In what follows, we provide more details about each existing method along with 
its performance in practice, which explains the motivation behind the  
new tool SimBa. 

\subsection{Sparse Rips filtration (inclusions)}

Let $P$ be a set of points in a metric space $({\mathcal M}, d)$. A {\em greedy permutation} $\{ p_1,..,p_n \} $ of $P$ is defined recursively as follows: 
Let $p_1 \in P$ be any
point and define $p_i$ recursively as 
$p_i = \argmax_{p \in P \setminus P_{i-1}} d(p,P_{i-1}),$
where $P_{i-1} = \{ p_1,...,p_{i-1}\}$. 
This gives rise to a nested sequence of subsets $P_1 \subset P_2 \subset \cdots P_n=P$. 
Furthermore, each subset $P_i$ is locally dense and uniform (net) in the
following sense.
Define the insertion radius $\lambda_{p_i}$ of a point $p_i$ as $\lambda_{p_i} = d(p_i,P_{i-1})$. Each $P_i$ is a $\lambda_{p_i}$-net of $P$, meaning that $d(p,P_i) \le \lambda_{p_i}$ for every $p \in P$ and $d(p,q) \ge \lambda_{p_i}$ for every distinct pairs $p,q \in P_i$. 
These nets can be extended to a single-parameter family of nets as $\{ N_\gamma \}$ where $N_{\gamma} = \{p \in P | \lambda_p > \gamma\}$ is a $\gamma$-net of $P$. 

Using the idea of Sheehy~\cite{S2012},
Buchet et al.~\cite{BCOS15} define a Rips-like filtration using the above specific nets
and assigning weights to points whose geometric interpretation is explained
in~\cite{CJS2015}.
Each point $p\in P$ is associated
with a weight $w_p(\alpha)$ at scale $\alpha$ as 
\begin{equation*}
	w_p(\alpha) = \begin{cases}
		0 & \text{if}\  \alpha \le \frac{\lambda_p}{\eps} \\
		\alpha-\frac{\lambda_p}{\eps} & \text{if}\  \frac{\lambda_p}{\eps} < \alpha \le \frac{\lambda_p}{\eps(1-\eps)} \\
		\eps \alpha & \text{if}\  \frac{\lambda_p}{\eps(1-\eps)} \le \alpha
	\end{cases}
\end{equation*}
where $0 < \eps < 1$ is an input constant that controls the sparsity of the filtration. Then, the perturbed distance between pairs of points is defined as
\begin{equation*}
	\dpert_{\alpha}(p,q) = d(p,q) + w_p(\alpha) + w_q(\alpha).
\end{equation*}
Using the perturbed distance $\dpert_\alpha$, the Sparse Rips complex at scale $\alpha$ is defined as 
\begin{equation*}
	\QRips^\alpha = \{\sigma \subset {N}_{\eps(1-\eps)\alpha}\ |\ \forall p,q \in \sigma,\ \dpert_\alpha(p,q) \le 2\alpha \}.
\end{equation*}
The sequence of $\{\QRips^\alpha\}_{\alpha > 0}$ does not form a nested sequence of spaces because
the vertex set of each $\QRips^\alpha$ comes from the net ${N}_{\eps(1-\eps)\alpha}$ and may decrease as $\alpha$ increases. 
However, one can take $\SRips^\alpha = \bigcup_{\alpha' \le \alpha} \QRips^{\alpha'}$ and build a natural filtration $\{ \SRips^\alpha \hookrightarrow \SRips^{\alpha'} \}_{\alpha < \alpha'}$ connected by inclusions. 
It is shown that the persistence barcodes of the filtration $\{\SRips^\alpha \}_{\alpha > 0}$ approximate those of the Rips filtration $\{\Rips^\alpha \}_{\alpha > 0}$ \cite{S2012}. 

We use the code from \cite{BCOS15} to compute this sparse Rips filtration $\{\SRips^\alpha \}_{\alpha > 0}$.
We then use GUDHI \cite{GUDHI2015} to compute its persistent barcodes as GUDHI has the state-of-the-art performance for handling large complexes due to a compression technique~\cite{BDM2013} for inclusion-based filtrations. 

As a common test case to illustrate the performance of various existing methods, we use a 3-dimensional point set sampled from a surface model called MotherChild; see Figure \ref{fig:FE}. We choose this model because the ground truth is available and also because existing methods have trouble (to different degrees) handling high-dimensional data. 
The size of the point set is $23075$. 
For indicating memory consumption, we refer to {\em cumulative
size} which is the total number of simplices arising in the filtration,
and also to {\em maximum size} which is the maximum over all complexes
arising in the filtration. For Sparse Rips filtrations two sizes 
coincide at the last complex due to inclusions.
Figure~\ref{fig:FE_eps_size} shows the cumulative size
with different 
Sparse Rips parameter $\eps$. It is minimum when $\eps$ is between 
$0.8$ and $0.9$. So, we choose $\eps = 0.8$ to achieve the best performance 
while observing that the approximation quality does not suffer much 
as predicted by the theory. 

The original persistence barcodes are shown in Figure \ref{fig:FE_a_pers_ori}. 
Since it becomes hard to see the main (long) bars in presence
of all spurious ones creating excessive overlaps, we remove all short bars whose ratio between death and birth time is smaller than a threshold for $1$-dimensional
homology group $H_1$. The bars for $H_0$ and $H_2$ are not denoised.
Unless specified otherwise, all barcodes are denoised in the 
same way. The denoised barcodes are shown in Figure \ref{fig:FE_a_pers_den}, 
one for $H_0$, four short and four long for $H_1$ (MotherChild has genus 4),
and one for $H_2$. The cumulative size of the Sparse Rips complexes 
in the filtration is 43.5 million and the total time cost is about 350 seconds.

\begin{figure}[h]
	\centering
	\begin{subfigure}[b]{0.32\textwidth}
		\centering\includegraphics[height=0.8\textwidth]{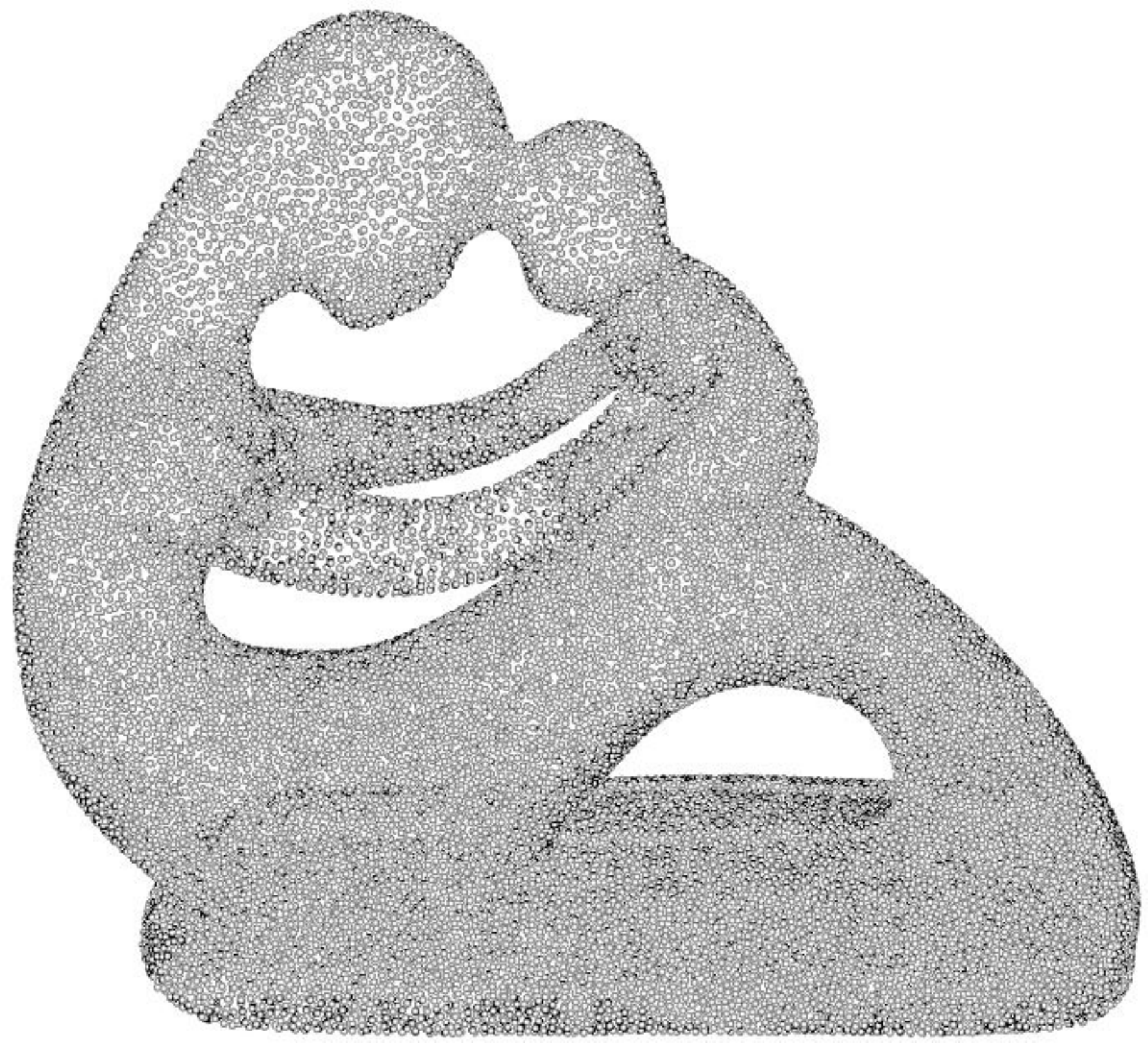}
		\caption{MotherChild model}\label{fig:FE}
	\end{subfigure}
	\begin{subfigure}[b]{0.32\textwidth}
		\centering\includegraphics[height=0.8\textwidth]{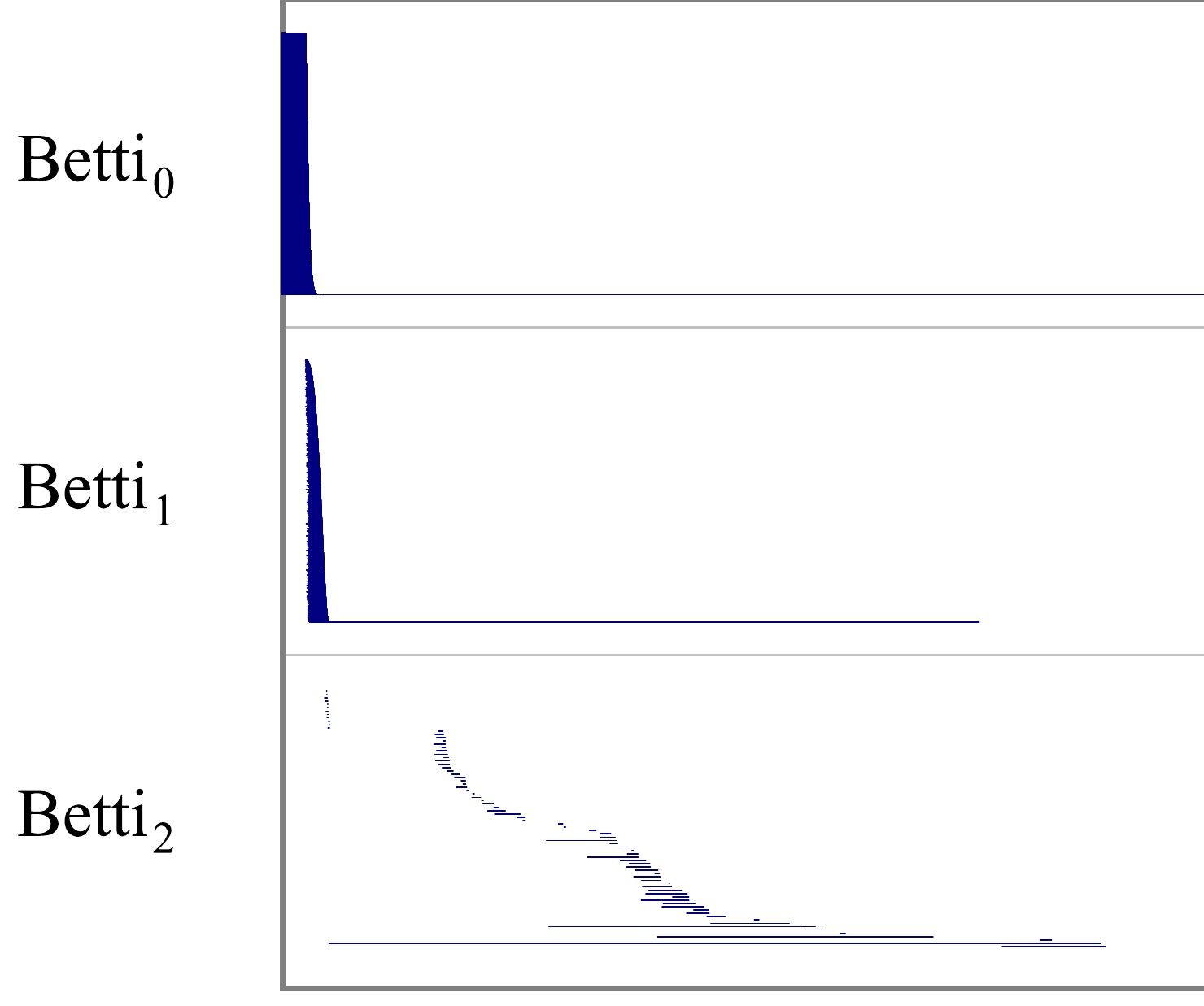}
		\caption{S.R. + GUDHI (original)}\label{fig:FE_a_pers_ori}
	\end{subfigure}
	\begin{subfigure}[b]{0.32\textwidth}
		\centering\includegraphics[height=0.8\textwidth]{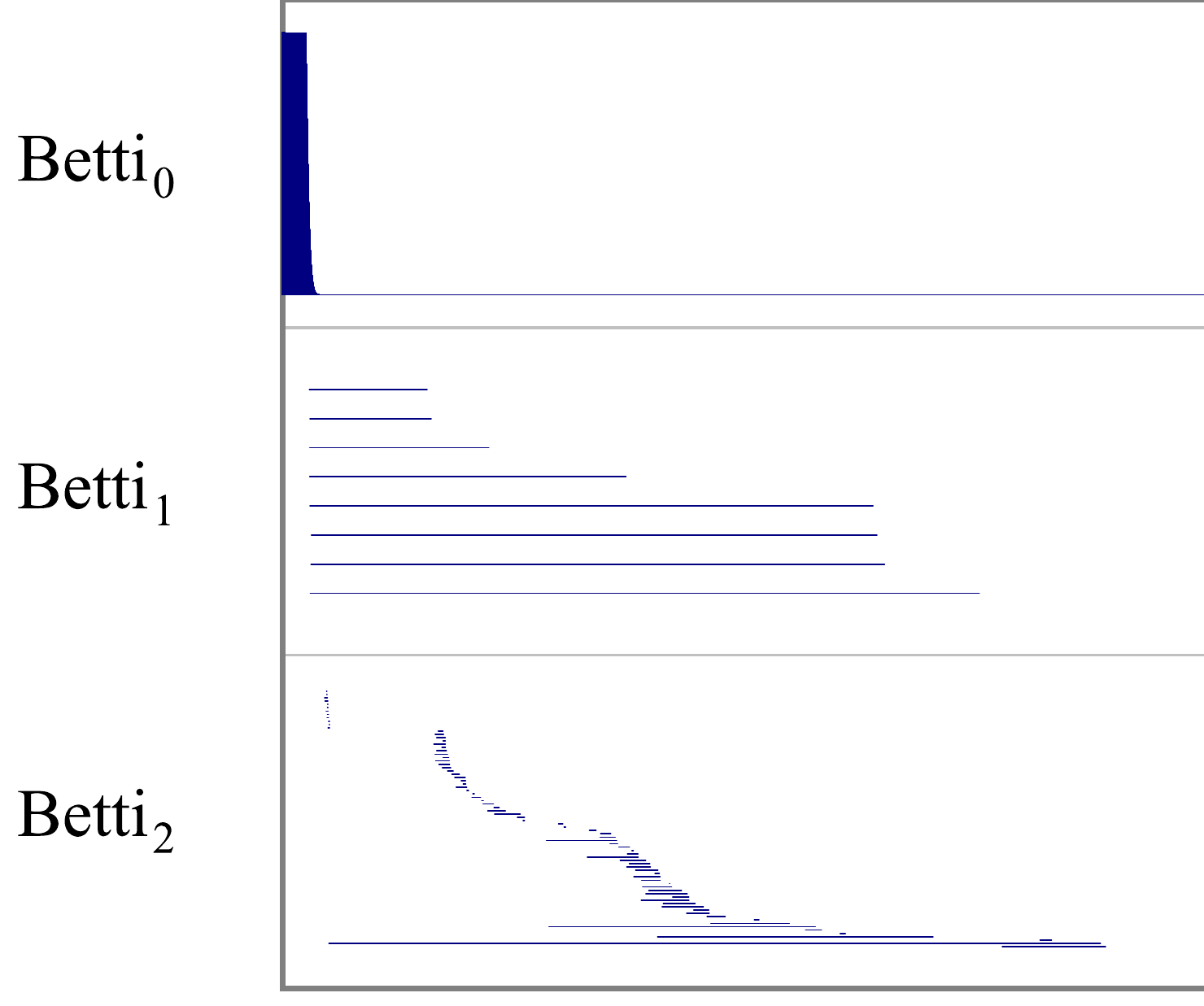}
		\caption{S.R. + GUDHI (denoised)}\label{fig:FE_a_pers_den}
	\end{subfigure}
	
	\begin{subfigure}{0.32\textwidth}
		\centering\includegraphics[height=0.8\textwidth]{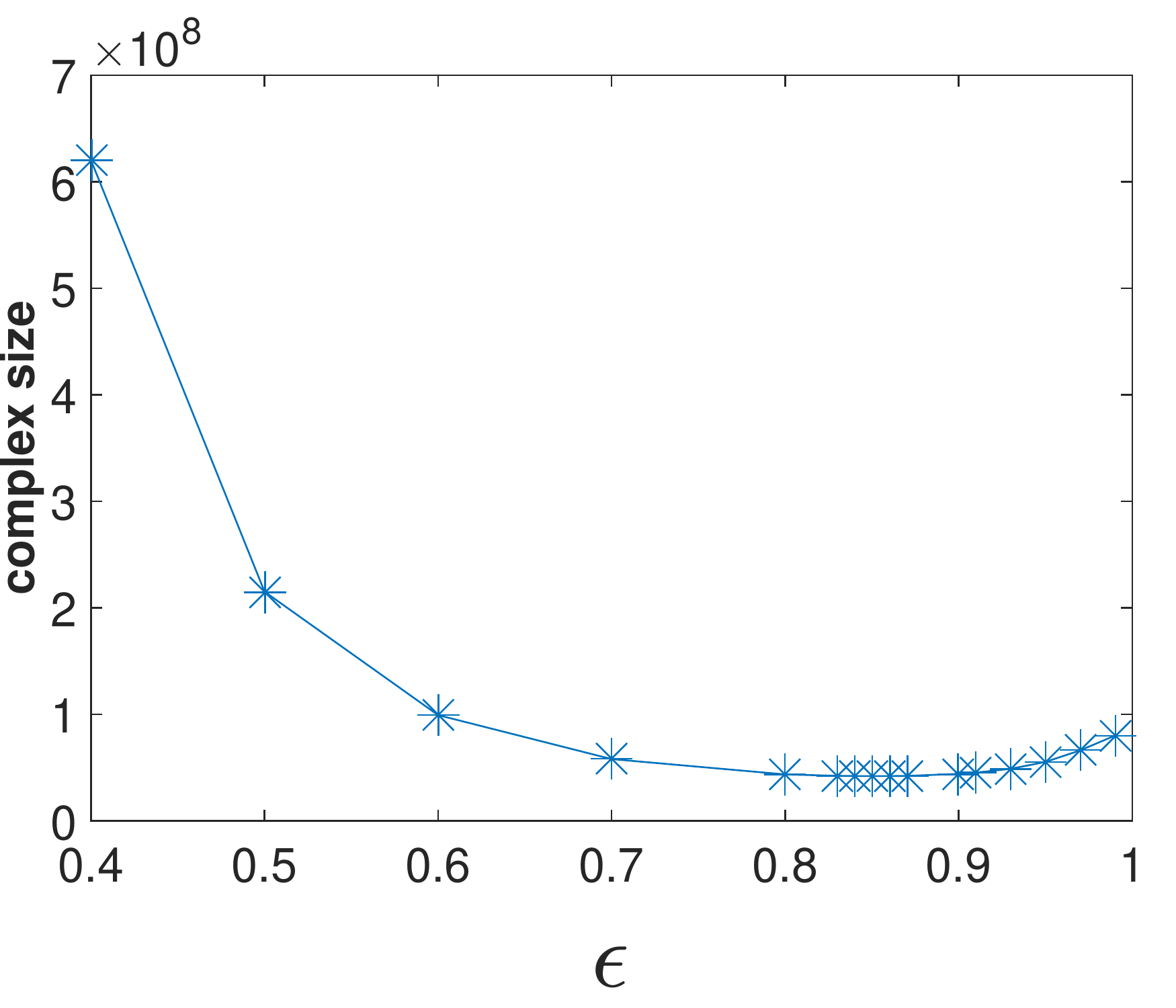}
		\caption{cumulative size}\label{fig:FE_eps_size}
	\end{subfigure}
	\begin{subfigure}{0.32\textwidth}
		\centering\includegraphics[height=0.8\textwidth]{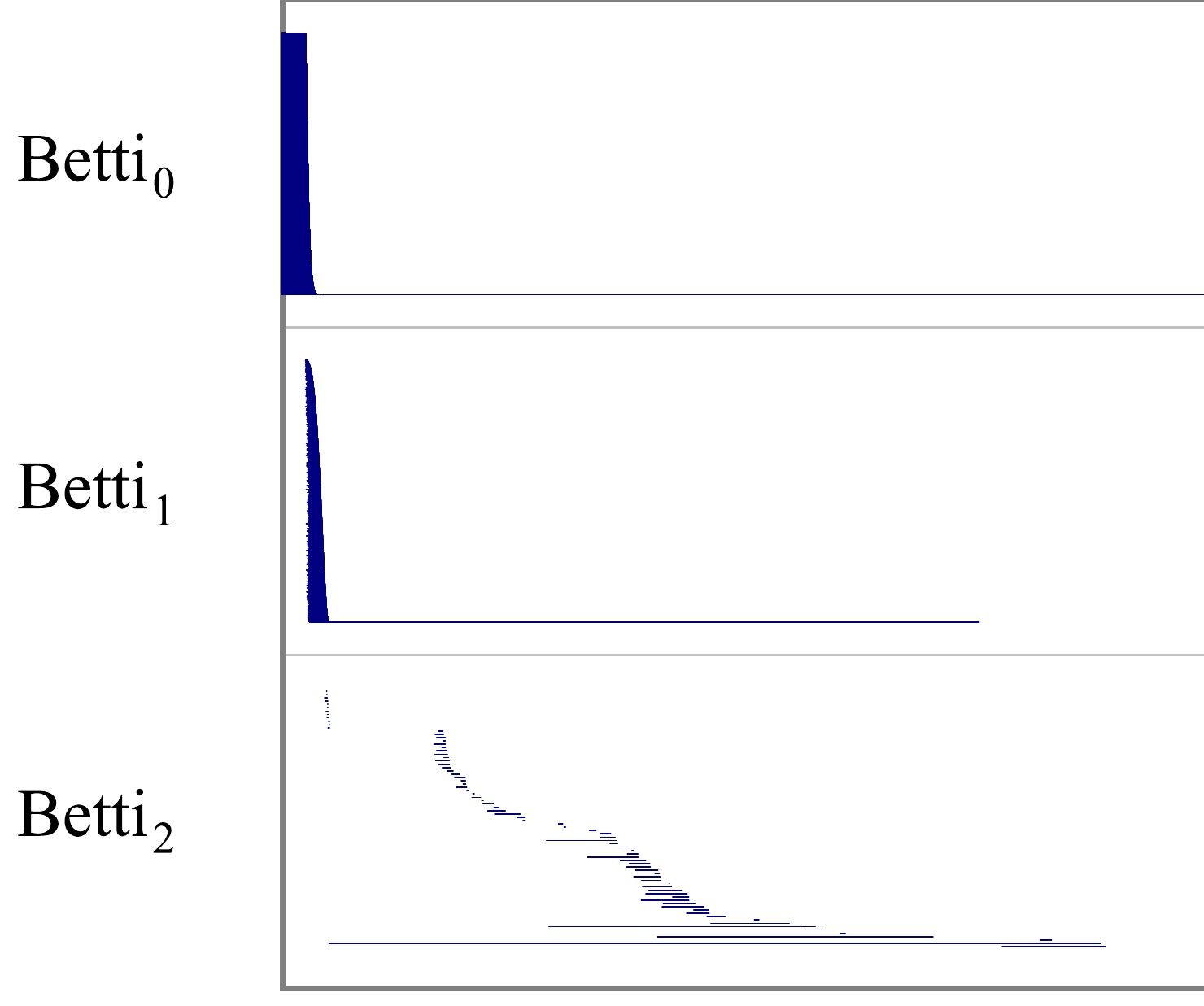}
		\caption{S.R. + Simpers (original)}\label{fig:FE_b_pers_ori}
	\end{subfigure}
	\begin{subfigure}{0.32\textwidth}
		\centering\includegraphics[height=0.8\textwidth]{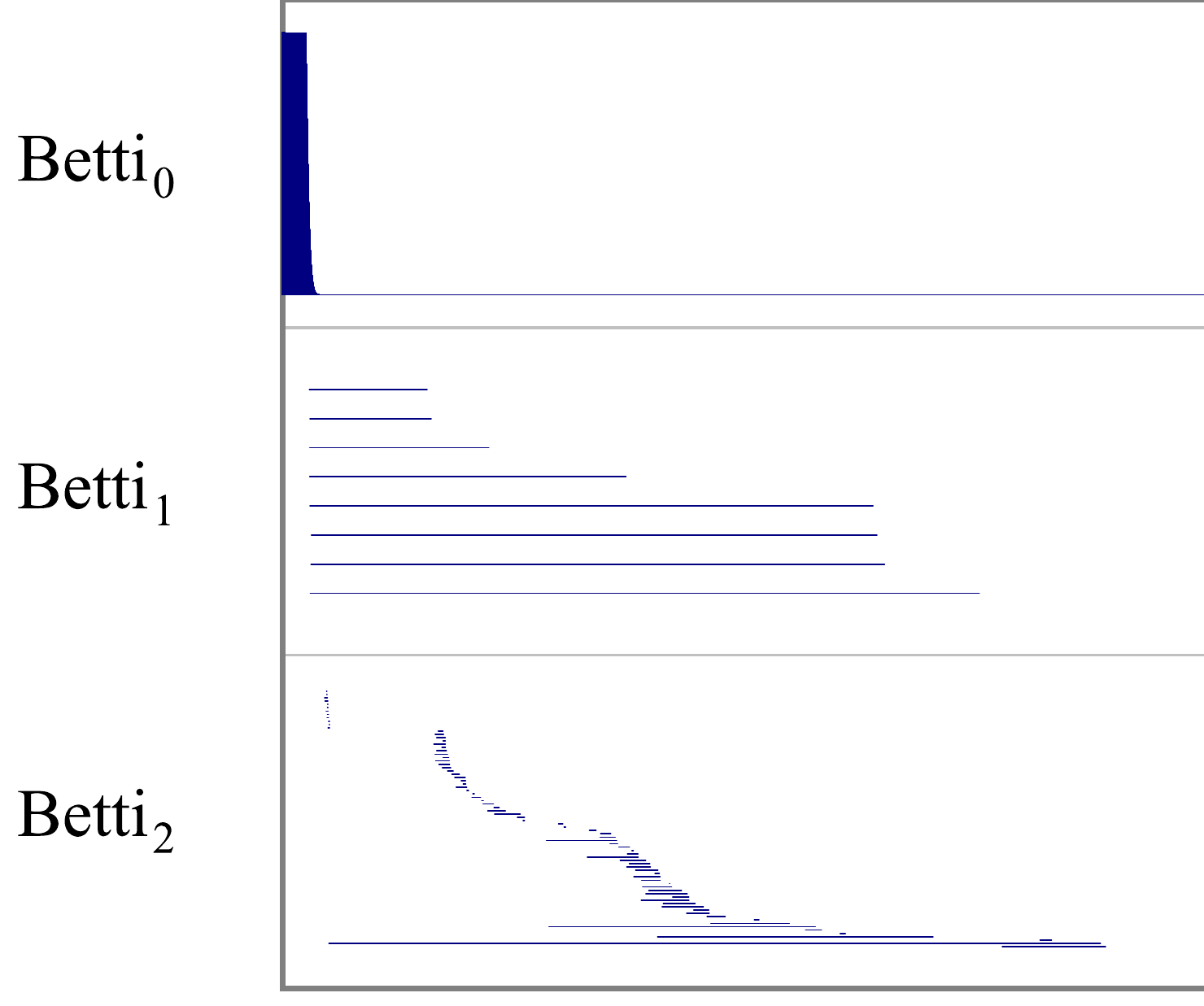}
		\caption{S.R. + Simpers (denoised)}\label{fig:FE_b_pers_den}
	\end{subfigure}
	\caption{MotherChild surface model and its persistence barcodes computed by Sparse Rips (S.R.) based approaches. Since the surface has genus $4$, its barcodes contain long bars: $1$ for $H_0$, $8$ for $H_1$, and $1$ for $H_2$. The minimum cumulative size for complexes, which is about 43 million,
		is achieved around $\eps = 0.8$} 
	\label{fig:FE_ab}
\end{figure}

\subsection{Sparse Rips with collapse}
The persistence barcodes for the inclusion-based filtration $\{ \SRips^\alpha\}_{\alpha > 0}$ are the same as the barcodes of the filtration $\{ \QRips^{\alpha}\}_{\alpha > 0}$ connected by simplicial
maps $\QRips^\alpha \to \QRips^{\alpha'}$ for $\alpha < \alpha'$. 
Specifically, these simplicial maps originate from vertex
collapses defined by the following projection map: 
\begin{equation*}
	\mu_\alpha(p) = \begin{cases}
		p & \text{if}\ p \in N_{\eps(1-\eps)\alpha} \\
		\underset{q \in N_{\eps\alpha}}{\argmin}\ d(p,q) & \text{otherwise}
	\end{cases}
\end{equation*}  
For any scale $\alpha$, the projection $\mu_\alpha$ maps the points of $P$ to the net $N_{\eps(1-\eps)\alpha} \supseteq N_{\eps\alpha}$. One can view it as
$p$ being deleted at time (scale) $\alpha_p =\frac{\lambda_p}{\eps(1-\eps)}$. 
We can construct the sequence of Sparse Rips complexes
$\{\QRips^\alpha\}_{\alpha>0}$ connected with simplicial maps 
induced by insertions and vertex collapses as $\alpha$ increases: 
Specifically, we delete the vertex $p$ and all its incident simplices by collapsing it to its projection $\mu_{\alpha_p}(p)$ where $\alpha_p = \frac{\lambda_p}{\eps(1-\eps)}$, when entering complex $\QRips^{\alpha_p}$. 
See ~\cite{CJS2015} for more details. 

In this approach we need to compute the persistence induced with simplicial
maps. For this, we use the only available software Simpers \cite{Simpers2015} based on the
algorithm presented in~\cite{DFW2014}. 
Our experimental results on the MotherChild model of
Figure~\ref{fig:FE} with $\eps = 0.8$ are given in Figure~\ref{fig:FE_b_pers_ori} and \ref{fig:FE_b_pers_den}. The barcodes are exactly the same as those in Figure \ref{fig:FE_a_pers_ori} and \ref{fig:FE_a_pers_den}. The cumulative size 
of the entire sequence is the same, $43.5$ million, because the final
complex in $\SRips^\alpha$ is the union of all complexes in $\{\QRips^\alpha\}_{\alpha > 0}$. 
However, the maximum size in the sequence is $24.9$ million due to 
vertex collapses in contrast to the maximum size for $\{ \SRips^\alpha\}_{\alpha > 0}$ which equals the cumulative size. 
The time cost for this approach 
is $463.7$ seconds which is larger than that for Sparse Rips with GUDHI. 
So compared to Sparse Rips with GUDHI, this approach has smaller 
maximum size due to collapse but costs more time for 
computing persistence since Simpers computes persistence over collapses 
which are slower operations than inclusions.

While the size of these Sparse Rips complexes is linear in the
number of input points, the hidden constant factor depends exponentially on the doubling dimension of the metric space where points are sampled from. Empirically, we note that the size is still large, and becomes much worse as the dimension of data increases. 
For example, for the Gesture Phase data in Table~\ref{fig:table} which has only $1747$ points in $\mathbb{R}^{18}$, the cumulative size of the Sparse Rips filtration is $45.6$ million, which approaches the limit GUDHI or Simpers can handle. For other larger data sets such as Primary Circle or Survivin, the complex reaches a size beyond this limit. Moreover, one has to pre-compute a greedy permutation of the input point set before constructing the Sparse Rips filtration. This computation is usually costly requiring 
furthest point computations for which  
software as efficient as ANN (for nearest neighbors) is not available.
This motivates us to consider the batched approach considered next.

\subsection{Batch-collapsed Rips}
For handling large and high dimensional
data, we need a more aggressive sparsification than the Sparse Rips filtration. 
We consider the Batch-collapsed Rips filtration, which has been proposed 
previously in \cite{DFW2014} (section 6.1). 

Given a set of points $P$, first set $V_0 := P$ and compute the shortest pairwise distance $\alpha$. We next construct a sequence of vertex sets $V_k, k \in [0,m]$ such that $V_{k+1}$ is an $\alpha c^{k+1}$-net of $V_k$ where $c > 1$
is a parameter that controls the rate of the scale increase. Consider
the vertex map $\vmap_k : V_k \rightarrow V_{k+1}$, for $k \in [0,m-1]$, such that for any $v \in V_k$, $\vmap_k(v)$ is $v$'s nearest neighbor in $V_{k+1}$. It can be shown that each vertex map $\vmap_k$ induces a well-defined simplicial map $\oldhmap_k: \Rips^{\alpha c^k \frac{3c-1}{c-1}}(V_k) \rightarrow \Rips^{\alpha c^{k+1} \frac{3c-1}{c-1}}(V_{k+1})$. 
The \emph{Batch-collapsed Rips filtration} is: 
\begin{align}
	\label{bcRips}
	\Rips^0(V_0) \xrightarrow{\oldhmap_0} \Rips^{\alpha c \frac{3c-1}{c-1}}(V_1) \cdots \xrightarrow{\oldhmap_{m-1}} \Rips^{\alpha c^m \frac{3c-1}{c-1}}(V_m).
\end{align}

Using the line of proof in~\cite{DFW2014}, one can show that the persistence of this sequence is a $3 \log(\frac{2}{c-1} + 3) $-approximation of the persistence diagram of Rips filtration given below.
\begin{align}
	\label{Ripsfiltration}
	\Rips^0(V_0) \hookrightarrow \Rips^{\alpha c}(V_0) \cdots \hookrightarrow \Rips^{\alpha c^m}(V_0).
\end{align}
The blowup in scale by the factor of $\frac{3c-1}{c-1}$ results from the
proof, which in practice causes some problems. We elaborate this further. 
To satisfy the approximation guarantee, one has to show that
the persistence modules arising from Batch-collapsed Rips in sequence (\ref{bcRips}) and the standard Rips in sequence (\ref{Ripsfiltration}) interleave. 
In particular, this requires that we have well-defined simplicial maps from complexes in sequence (\ref{bcRips}) to those in sequence (\ref{Ripsfiltration}) and vice versa. The multiplicative factor $\frac{3c-1}{c-1}$ is needed to ensure that there is a well-defined simplicial map  $\Rips^{\alpha c^k}(V_0) \rightarrow \Rips^{\alpha c^{k} \frac{3c-1}{c-1}}(V_{k})$, as $\Rips^{\alpha c^{k} \frac{3c-1}{c-1}}(V_{k})$ has to be sufficiently connected to include all the images of the simplices in the domain Rips complex $\Rips^{\alpha c^k}(V_0)$. 
The side effect of this is that the Batch-collapsed Rips complex has to be built at a much larger scale than the Rips complex, and it ends up with many unnecessary connections and thus more simplices in practice. This also causes a trade-off: Larger $c$ reduces the over-connection but results in a worse approximation factor leading to a worse approximation quality. It is not clear how to set an increase rate that achieves both good approximation quality and efficiency for a specific data set. 

We experimented Batch-collapsed Rips with Simpers on the same MotherChild model. Figure \ref{fig:FE_c_pers_13}, \ref{fig:FE_c_pers_15} and \ref{fig:FE_c_pers_20} show the persistence barcodes for different values of $c$. Observe that smaller values of $c$ give better approximation. The barcode for $c = 1.3$ is the most similar among the three to that of Sparse Rips filtration in Figure \ref{fig:FE_a_pers_ori} which is supposed to be more accurate theoretically. On the other hand, when $c$ grows more than $1.8$, it starts to lose some main bars in $H_1$ and noisy bars get longer in $H_2$. On the other hand, Figure \ref{fig:FE_cd_compare} shows that, as $c$ increases, both complex size and time cost decrease drastically. When $c = 2.0$, it only involves less than $216K$ simplices and takes time $9.4$ seconds while, although $c = 1.3$ gives more accurate barcode, its size ($22.5$ million) and time ($325$s) approach those of the Sparse Rips. This demonstrates the dilemma that Batch-collapsed Rips faces in practice. 
We address this issue in our new approach SimBa. In particular, when
$c \le 2$, SimBa performs better than Batch-collapsed Rips for both size and time as shown in Figure \ref{fig:FE_cd_compare} while capturing all main bars correctly as shown in Figure \ref{fig:FE_cd_pers}.

\section{SimBa}
\label{sec:SimBa}
To tame the over-connection in Batch-collapsed Rips, we replace 
the sequence in~(\ref{bcRips}) with the sequence 
below where the parameter does not
incur the extra factor $\frac{3c-1}{c-1}$: 
\begin{equation}
	\label{SimBafiltration}
	\SetRips^0(V_0) \rightarrow
	\SetRips^{\alpha c}(V_1) \rightarrow
	\cdots
	\SetRips^{\alpha c^m}(V_m)
\end{equation}
The complexes $\SetRips^{\alpha c^k}(V_k)$ are built on the same
vertex sequence $\{V_k\}$ as in Batch-collapsed Rips, but the
distances among the vertices of $V_k$ are replaced with a {\em set distance}
which allows us to avoid the over-connection. For two sets of points (clusters)
$A,B\subset P$, we define their set distance as 
$d(A,B)=\min_{a\in A, b\in B} d(a,b)$. The sets that we consider
are the pre-images of the vertices in $V_k$ under the composition of
projections $\pi_i$'s, namely, for a vertex $v\in V_k$, we consider
the set 
\begin{equation*}
B_v^k = \{p \in V_0\ |\ \hv_k (p) = v \} ~~\text{where}~~\hv_k : V_0 \to V_k~~\text{is defined as}~~\hv_k = \vmap_{k-1} \circ \cdots \circ \vmap_0. 
\end{equation*}
The complex $\SetRips^{\alpha c^k} (V_k)$ is simply the clique complex 
induced by edges $\{ (u, v) \in V_k \mid d(B_u^k,B_v^k)\leq \alpha c^k\}$. 
Observe that $d(u,v)\geq d(B_u^k,B_v^k)$
which ensures that the normal connection between $u$ and $v$ for a Rips
filtration at the respective scale is not missed by considering the set
distance while still avoiding the over-connection.

It turns out that each vertex map (nearest neighbor projection) $\vmap_k: V_k \to V_{k+1}$ induces a simplicial map $h_k: \SetRips^{\alpha c^k}(V_k) \to \SetRips^{\alpha c^{k+1}}(V_{k+1})$. Instead of recomputing the simplicial complex each time, we generate elementary insertion and collapse operations incrementally for each $h_k$ in three steps: (i) collapse each $v \in V_k \setminus V_{k+1}$ to its image $\vmap_k(v)$ in $V_{k+1}$ along with all incident simplices, (ii) insert new edges between two vertices in $V_{k+1}$ if the distance between the two sets they represent are smaller than or equal to the current scale, and (iii) insert all new clique simplices containing new edges generated by (i) and (ii). Each $h_k$ is processed in one batch, starting from a simplicial complex on vertices in $V_k$ and resulting in a simplicial complex on vertices in $V_{k+1}$. The collapse and insertions of new simplices are exactly what
Simpers need for computing the persistence. 

\subsection{Implementation Details}
The advantage of SimBa (and Batch-collapsed Rips) over Sparse Rips filtrations is mainly due 
to the batched approach, which requires us to compute $\delta$-nets of a point set for some $\delta$ repeatedly. 
Its advantage over the Batch-collapsed Rips is credited to the use of set distances. 
These computations require k-nearest neighbor search and fixed radius search for which efficient
library like ANN \cite{ANN2010} exists. We take advantage of this available software. 

To compute a $\delta$-net of a given point set (to obtain $V_{k+1}$ from $V_k$), we randomly pick an untouched point, say $p$, use fixed-radius search to find all points in the ball of radius $\delta$ around $p$, map them to it, and mark them processed. We do this repeatedly until there is no untouched point left. We observe that this sub-sampling 
procedure can be carried out faster at early stage when $\delta$ is small because those points whose nearest neighbor distances are larger than the current $\delta$ can be taken directly into the net--they are all mapped to themselves and no other points are mapped to them. So, we maintain a list $L$ of the points ordered by their nearest neighbor distances in increasing order and process them sequentially for $\delta$-net computations. To compute the net points $V_{k+1}$ from $V_k$, we carry out the full sub-sampling process only on the points in $V_k$ that are already known to have nearest neighbor distances below $\delta$ and the new ones that qualify from $L$ for increased $\delta$. 
After $\delta$ becomes more than the largest nearest neighbor distance, 
we convert to the usual net computation.

Next, we describe an efficient implementation of the set distance computation, which being a basic operation in SimBa, speeds it up significantly. A straightforward implementation requires quadratic time, but we can make it more efficient in practice with the help of the ANN library. We use a hybrid strategy as follows. The sets $B_u^k$ for vertices $u\in V_k$ are maintained by a union-find data structure. As vertices are collapsed while going from $V_k$ to $V_{k+1}$,
the sets of the collapsed vertices are merged to that of the target vertex.
At early stages, when the number of sets (i.e, the size of $V_k$) is large and the diameter of each set is potentially small, we avoid computing set distances for all pairs. 
For each processing set $B_u^k$, we only need to find all the sets $B_v^k$ whose distances to $B_u^k$ are smaller than the current scale $\alpha' = \alpha c^k$. 
If so, we add an edge between $u$ and $v$. To find all these nearby sets, we can do a fixed-radius search in $V_0 = P$ around each point in $B_u^k$ within $\alpha'$ distance. For each point $v$ returned by the search, we find $v$ in the union-find data structure to identify its image $\hv_k(v) \in V_k$. If the representing set of $v$ is different from that of $u$, we add the edge $\hv_k(u) \hv_k(v)$. 

Later when $\alpha'$ becomes large, it may not be efficient to continue this fixed-radius search, as the number of candidate points from $P$ may be too large (can be $n$ in the worst case). 
So we fall back on pairwise set distance computation. 
In particular, when the cardinality of $V_k$ becomes lower than a threshold, say $1/10$ of the number of input points, we compute a pairwise set distance matrix (of size $|V_k| \times |V_k|$) among the surviving sets once and then keep updating the matrix with batch collapse thereafter. In particular, note that given sets $A, B, $ and $C$, the set distance $d( A \cup B, C) = \min \{ d(A, C), d(B, C) \}$. 

\subsection{SimBa on MotherChild model}

We compare SimBa with other approaches on the same MotherChild model. Figure \ref{fig:FE_d_pers_13}, \ref{fig:FE_d_pers_15} and \ref{fig:FE_d_pers_20} show the persistence barcodes computed by SimBa with different values of 
$c$. We see that SimBa captures all the main 0, 1, 2-dimensional bars for all values of $c$ in the range
from $1.3$ to $2.0$ as opposed to Batch-collapsed Rips which fails to capture the main $H_1$ bars for $c > 1.8$. 
It tolerates larger range of $c$ and thus is more robust than Batch-collapsed Rips. As expected, larger values of $c$ produce less bars since there are less batches. So, in practice, we should choose smaller $c$, say less than $1.5$. More importantly, as Figure \ref{fig:FE_cd_compare} shows, the size and time for SimBa are also stable against different values of $c$, all less than $100K$ simplices and $10$ seconds respectively for $c \le 2$. These are less than those for Batch-collapsed Rips and significantly less than those for Sparse Rips: In particular, when $c=1.3$, the maximum size for SimBa is $100K$, similar to when $c=2$. However, for Batch-collapsed Rips, the maximum size is closer to that
of SimBa when $c=2$, and is $22.5$ and $1.4$ million when $c=1.3$ and $c=1.5$
respectively. This size difference becomes even more prominent for high dimensional data, as Table \ref{fig:table} shows.
Although the approximation quality of SimBa is slightly
worse than that of Sparse Rips based approaches, it does capture all 
the main bars, and more importantly, costs significantly less time. This advantage allows SimBa to process much larger high dimensional data sets which no previous approaches can handle, as we illustrate in section \ref{sec:Exp}.

\begin{figure}[h]
	\centering
	\begin{subfigure}{0.32\textwidth}
		\centering\includegraphics[height=0.8\textwidth]{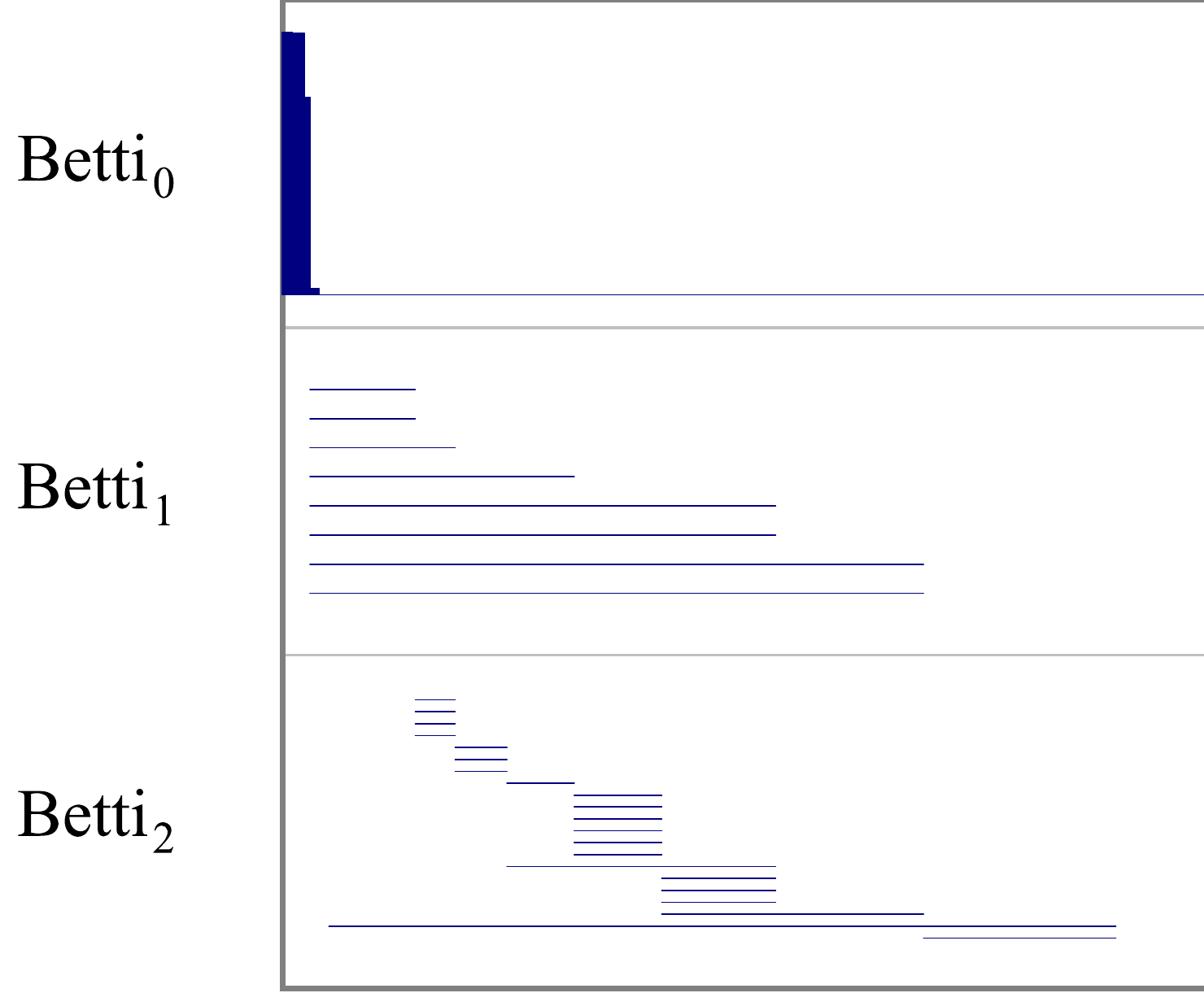}
		\caption{B.R. (c = 1.3)}\label{fig:FE_c_pers_13}
	\end{subfigure}
	\begin{subfigure}{0.32\textwidth}
		\centering\includegraphics[height=0.8\textwidth]{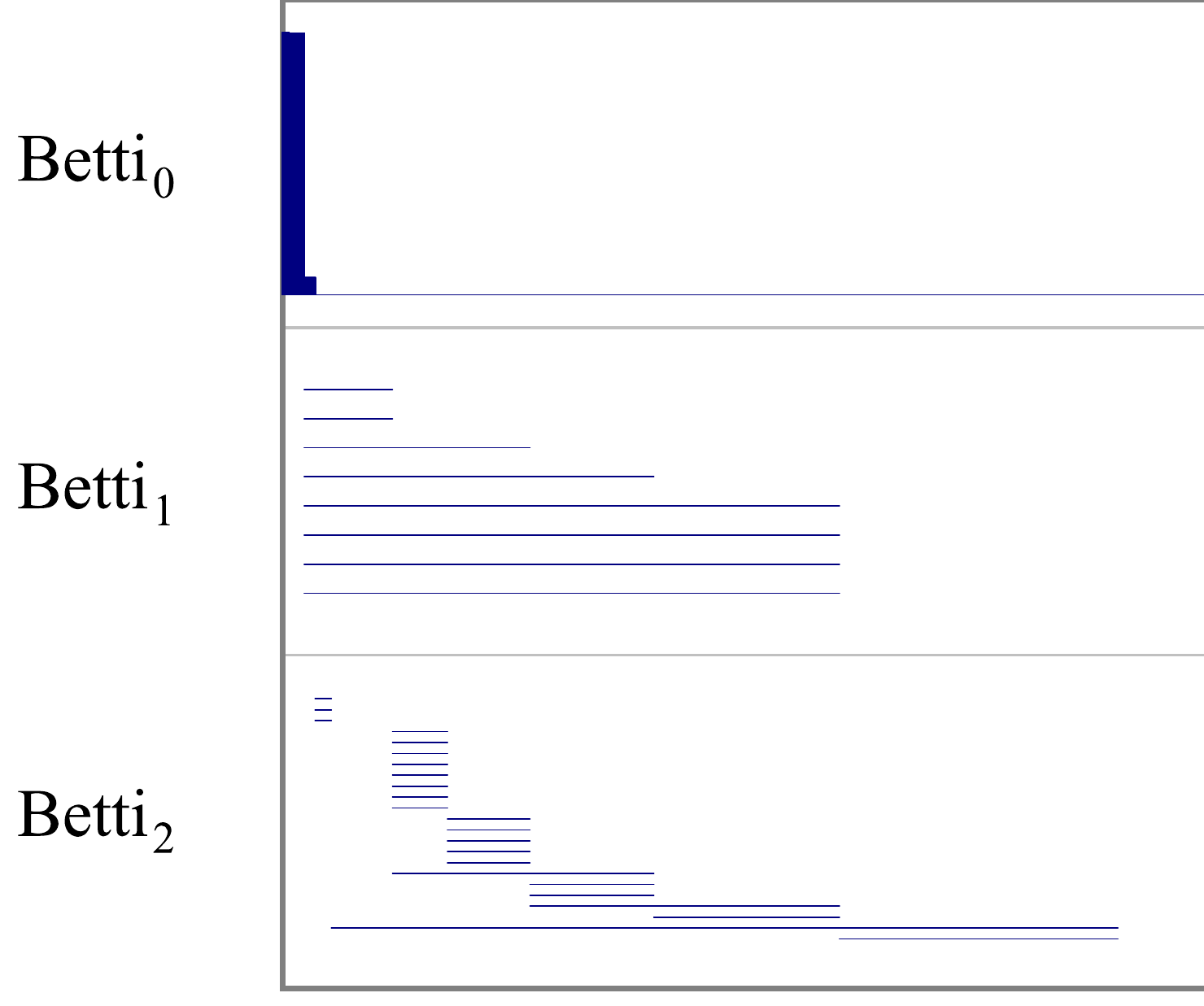}
		\caption{B.R. (c = 1.5)}\label{fig:FE_c_pers_15}
	\end{subfigure}
	\begin{subfigure}{0.32\textwidth}
		\centering\includegraphics[height=0.8\textwidth]{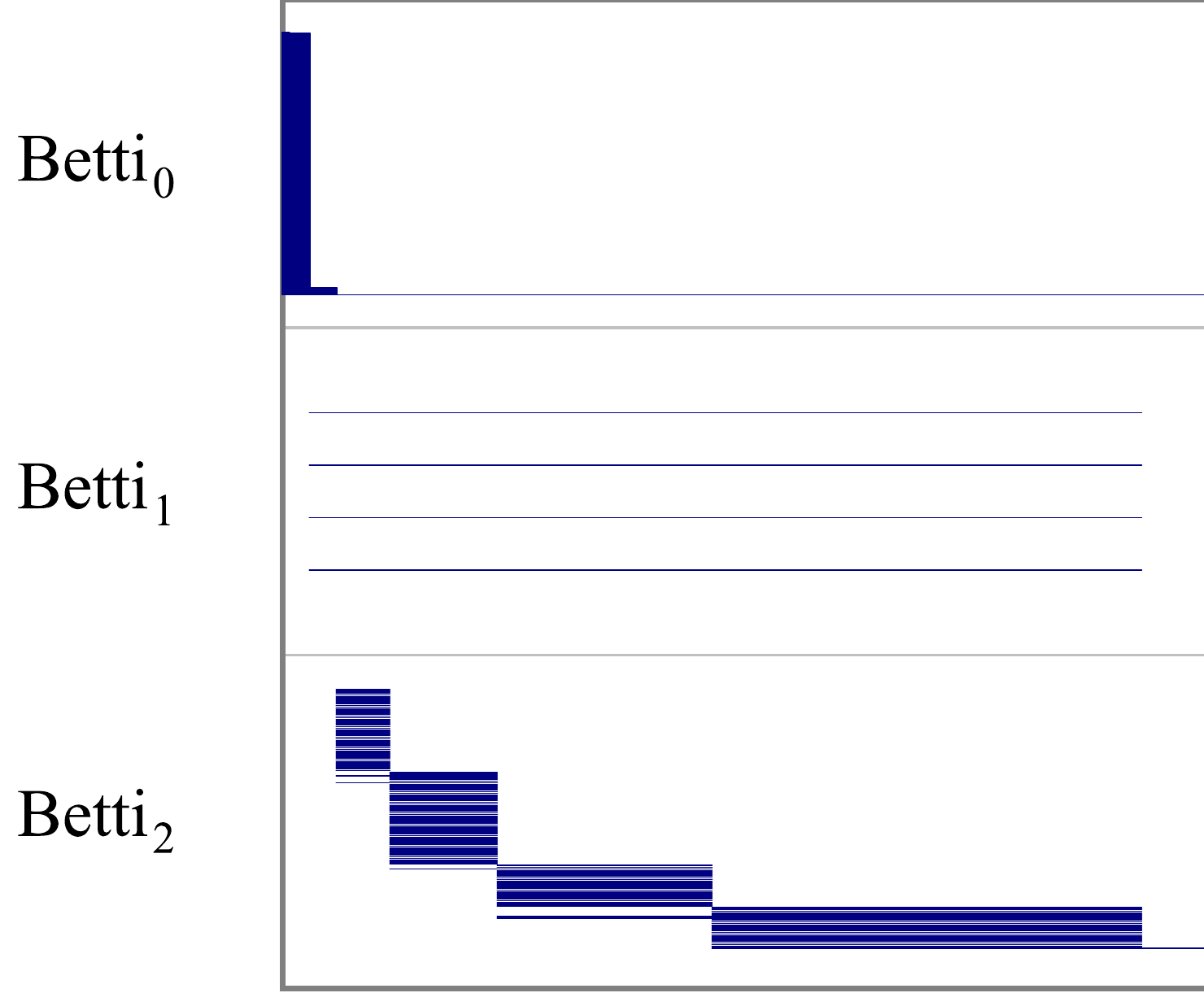}
		\caption{B.R. (c = 2.0)}\label{fig:FE_c_pers_20}
	\end{subfigure}
	
	\begin{subfigure}{0.32\textwidth}
		\centering\includegraphics[height=0.8\textwidth]{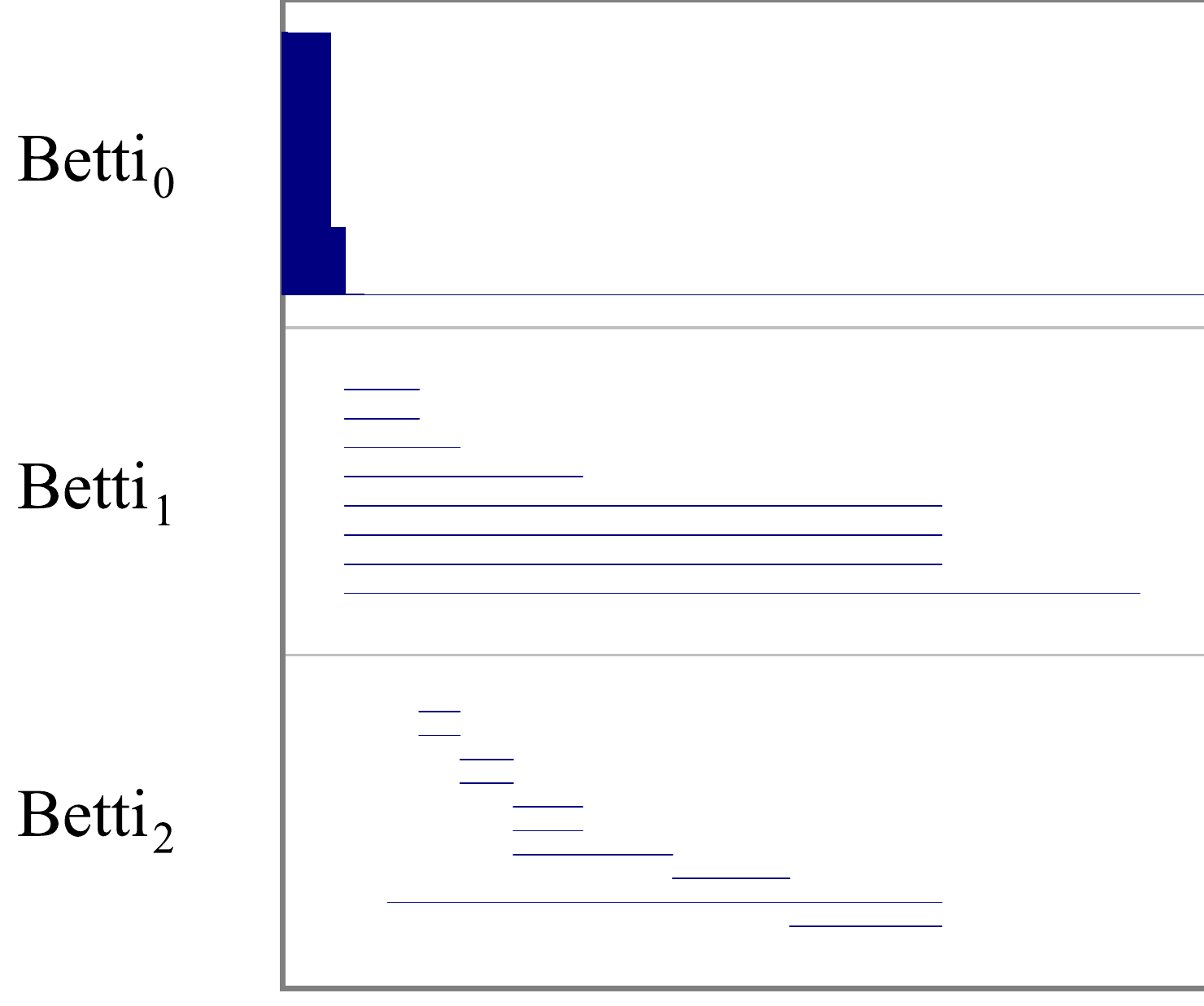}
		\caption{SimBa (c = 1.3)}\label{fig:FE_d_pers_13}
	\end{subfigure}
	\begin{subfigure}{0.32\textwidth}
		\centering\includegraphics[height=0.8\textwidth]{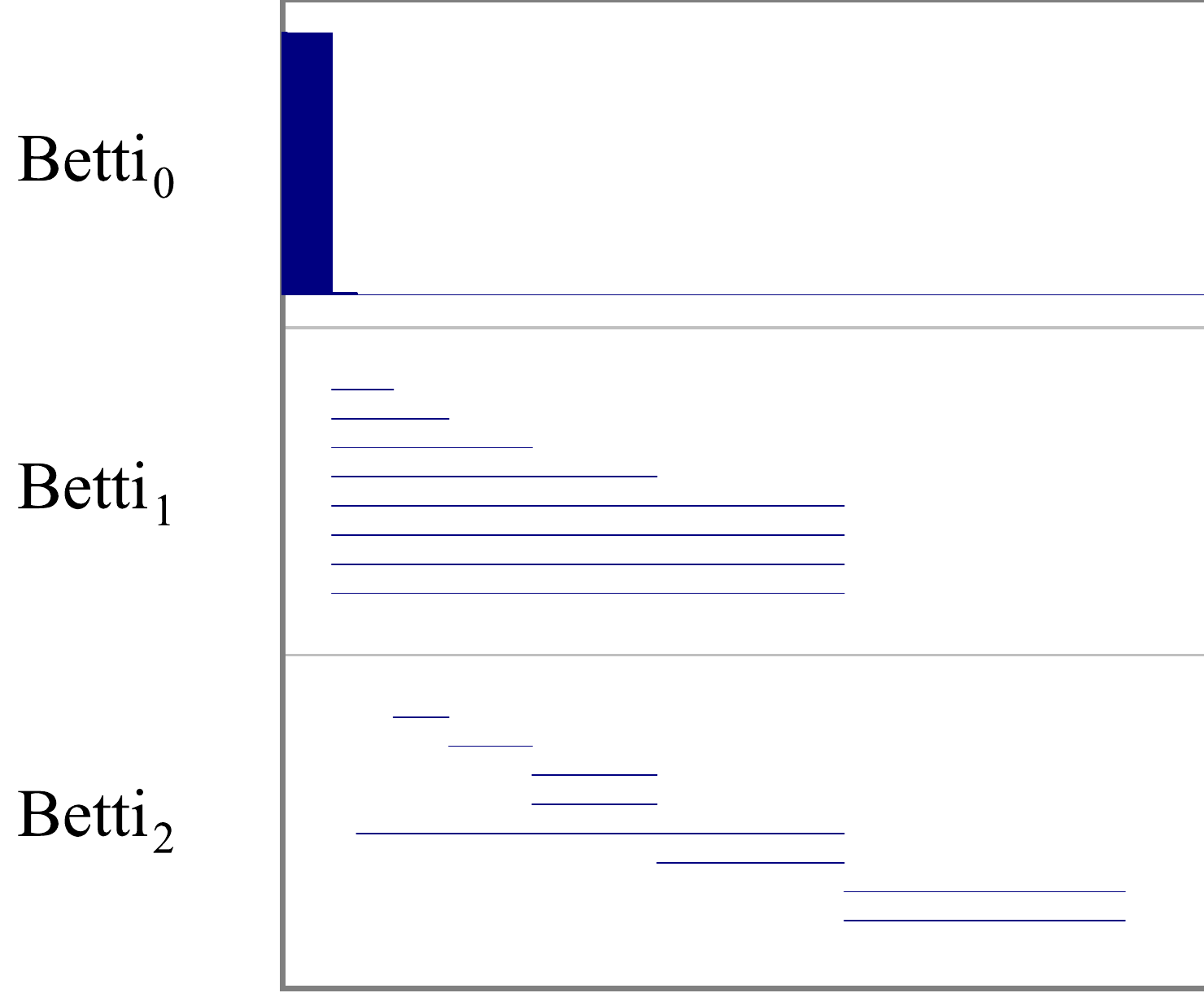}
		\caption{SimBa (c = 1.5)}\label{fig:FE_d_pers_15}
	\end{subfigure}
	\begin{subfigure}{0.32\textwidth}
		\centering\includegraphics[height=0.8\textwidth]{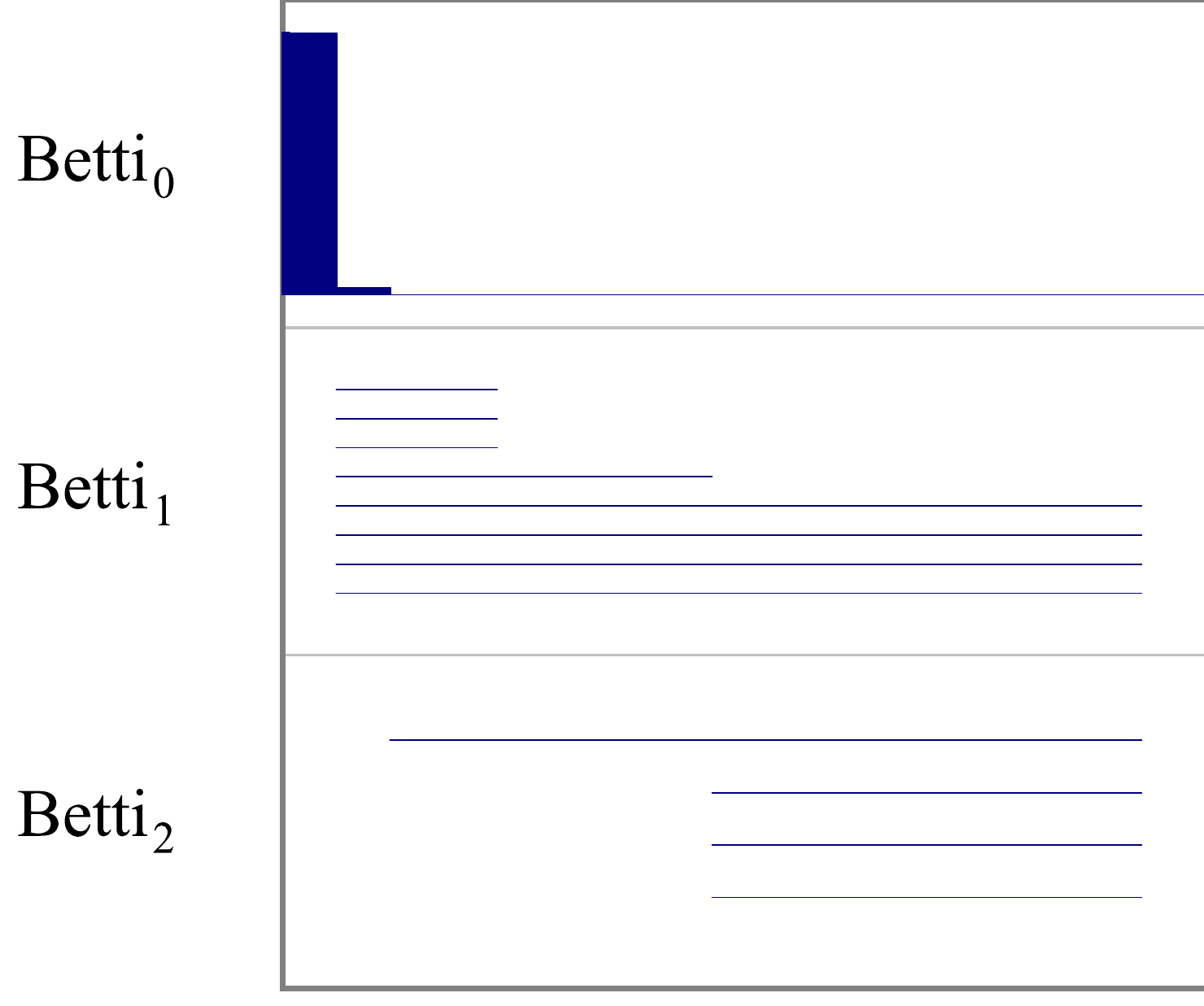}
		\caption{SimBa (c = 2.0)}\label{fig:FE_d_pers_20}
	\end{subfigure}
	
	\caption{Persistence barcodes computed by Batch-collapsed Rips plus Simpes (B.R.) and SimBa on the same MotherChild model. B.R. captures main bars for $H_1$ correctly for smaller values of $c$ as shown in Figure (a) and (b) and loses some for $c = 2.0$ as shown in Figure (c), while SimBa works for $c = 2.0$.}
	\label{fig:FE_cd_pers}
\end{figure}
\begin{figure}
	\centering
	\begin{subfigure}{0.48\textwidth}
		\centering\includegraphics[height=0.8\textwidth]{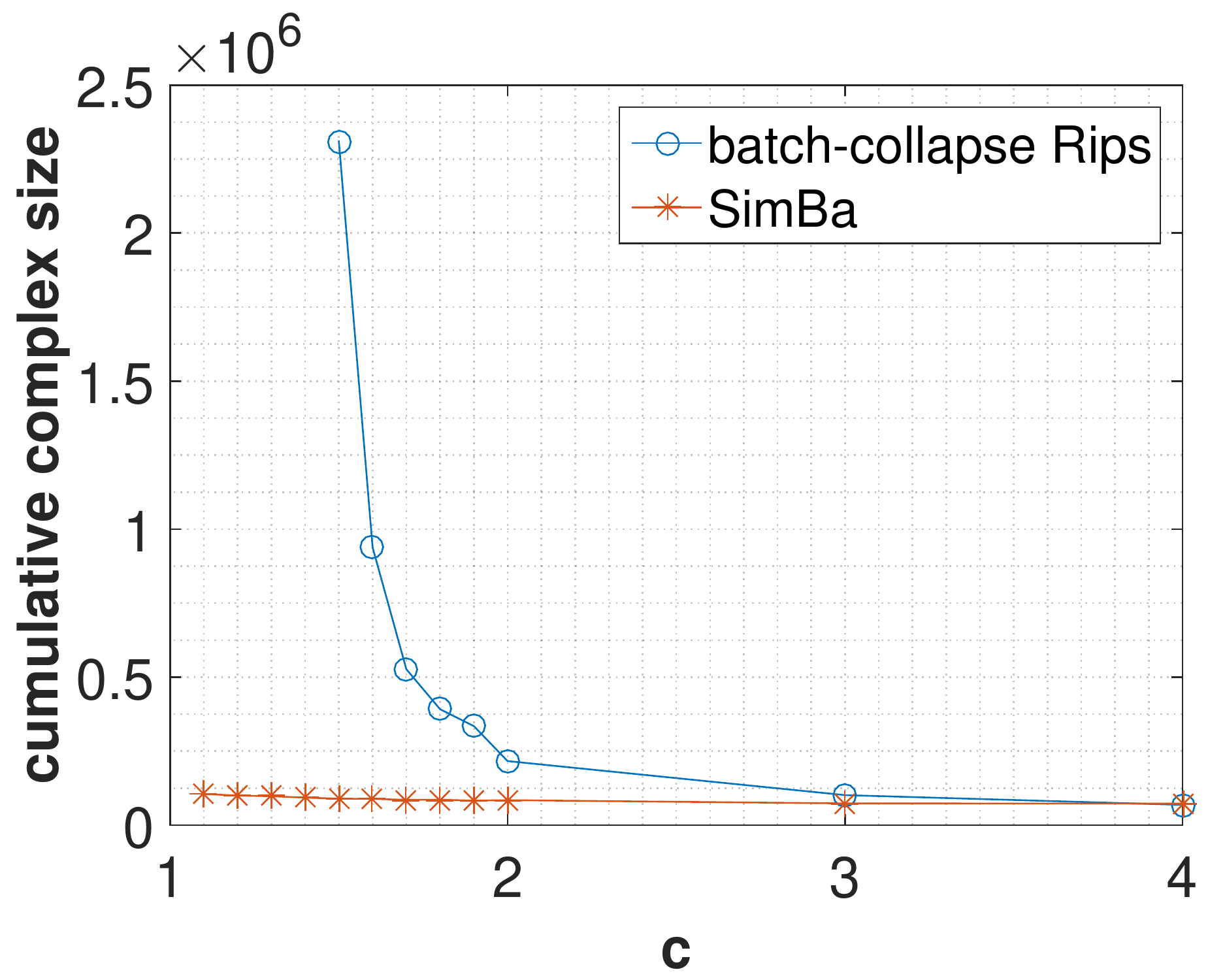}
		\caption{cumulative size}\label{fig:FE_cd_accusize}
	\end{subfigure}
	\begin{subfigure}{0.48\textwidth}
		\centering\includegraphics[height=0.8\textwidth]{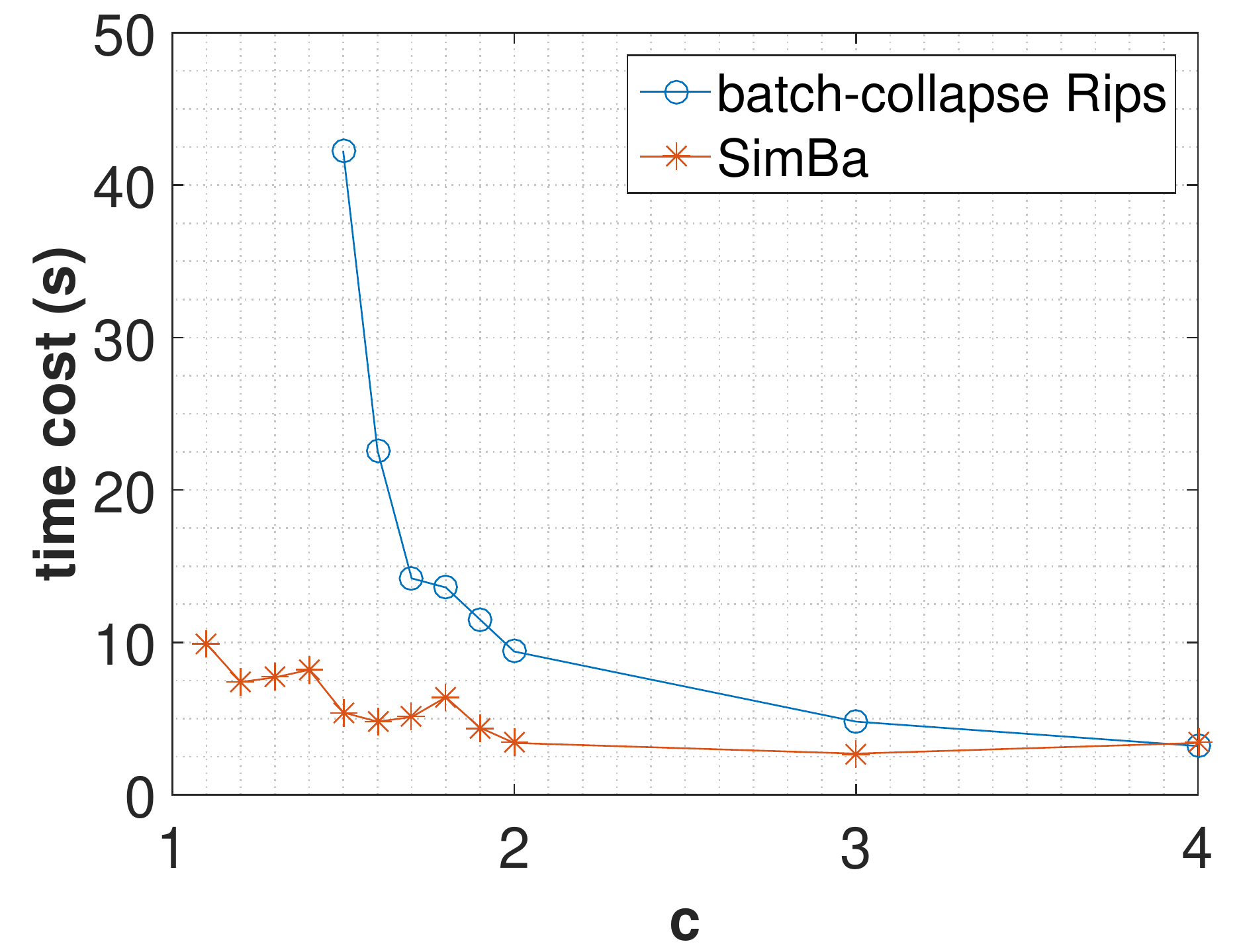}
		\caption{time cost}\label{fig:FE_cd_time}
	\end{subfigure}
	\caption{Complex size and time cost comparison between Batch-collapsed Rips and SimBa. SimBa beats Batch-collapsed Rips for both size and time when $c \le 2$. For $c > 2$, the barcodes of both batch-based approaches become too coarse to be useful in practice.}
	\label{fig:FE_cd_compare}
\end{figure}

\section{Approximation guarantee of SimBa}
Recall that the simplicial complex $\SetRips^{\alpha}(V_k)$ appearing in SimBa's filtration 
is defined as:
\begin{equation*}
	\SetRips^{\alpha}(V_k) = \{\sigma \subset V_k\ |\ \forall u, v \in \sigma, d(B_u^k, B_v^k) \le \alpha \}.
\end{equation*}
We prove that the persistence barcodes of SimBa's filtration in sequence (\ref{SimBafiltration}) approximates those of the Rips filtration in (\ref{Ripsfiltration}) by showing that the persistence modules induced by these two sequences interleave. 

First, observe that each vertex map $\vmap_k$ induces a well-defined simplicial map $h_k : \SetRips^{\alpha c^k}(V_k) \rightarrow \SetRips^{\alpha c^{k+1}}(V_{k+1})$. Indeed, for any edge $\{u, v\}$ in $\SetRips^{\alpha c^k}(V_k)$, suppose $u' = \vmap_k(u), v' = \vmap_k(v)$, then $B_u^k \subset B_{u'}^{k+1}$ and $B_v^k \subset B_{v'}^{k+1}$. So we have
	$d(B_{u'}^{k+1},B_{v'}^{k+1}) \leq d(B_{u}^k,B_{v}^k) \leq \alpha c^k < \alpha c^{k+1}.$ 
Therefore $\{u', v'\}$ must be an edge in $\SetRips^{\alpha c^{k+1}}(V_{k+1})$ as well. Since each complex in SimBa's filtration is a clique complex determined by edges, every simplex in $\SetRips^{\alpha c^k}(V_k)$ has a well-defined image in $\SetRips^{\alpha c^{k+1}}(V_{k+1})$. Thus, each $h_k$ is well-defined.

Recall that the map $\hv_k : V_0 \rightarrow V_{k+1}$ is defined as $\hv_k(v) = \vmap_k\circ\cdots\circ\vmap_0(v)$, which tracks the image of any point in $V_0 = P$ during the batch collapse process. Observe that the vertex map $\hv_k$ also induces a simplicial map $\hh_k :\Rips^{\alpha c^k}(V_0) \rightarrow \SetRips^{\alpha c^{k+1}}(V_{k+1})$: specifically, for any edge $(u, v) \in \Rips^{\alpha c^k}(V_0)$ with $d(u,v) \le \alpha c^k$, it is easy to see that $d(B_{\hv_k(u)}^k, B_{\hv_k(v)}^k) \le d(u, v) \le \alpha c^k  < \alpha c^{k+1}$, implying that $(\hv_k(u), \hv_k(v))$ is an edge in $\SetRips^{\alpha c^{k+1}}(V_{k+1})$.  
The key observation is the following lemma. 
\begin{lemma}
	\label{claim:trianglescommute}
	Each triangle in the diagram below commutes at homology level, where $i_k$ and $j_k$ are induced by inclusions, $\hmap_{k, t}:= \hmap_{k+t-1} \circ \cdots \circ \hmap_k$, $c > 1$, $t \ge \log_c(\frac{2}{c-1}+3)$ and $t \in Z$. 
$$	\xymatrix @R=1.0pc @C=1.0pc  
{\Rips^{\alpha c^k}(V_0) \ \ar@{^{(}->}[rr]^-{i_k} \ar[d]^{\hh_{k}}
		&  & \  
		\Rips^{\alpha c^{k+t}}(V_0) \ \ar[d]^{\hh_{k+t}} \\
		\SetRips^{\alpha c^k}(V_k) \ \ar@{^{(}->}[rru]^{j_k} \  \ar[rr]^{\hmap_{k, t}}  
		&  & \  \SetRips^{\alpha c^{k+t}}(V_{k+t})  }
$$ 
\end{lemma}

\begin{proof}
	First, we prove that there is indeed an inclusion map $j_k : \SetRips^{\alpha c^k}(V_k) \hookrightarrow \Rips^{\alpha c^{k+t}}(V_0)$.  
	In particular, we show for each edge $(u, v)$ in $\SetRips^{\alpha c^k}(V_k)$, it's also an edge in $\Rips^{\alpha c^{k+t}}(V_0)$. Suppose the set distance $d(B_u^k, B_v^k)$ is achieved by the closest pair $(u_0, v_0)$ between the two sets where $u_0 \in B_u^k, v_0 \in B_v^k$. Then $d(B_u^k, B_v^k) = d(u_0, v_0) \leq \alpha c^k$. Since $V_{i+1}$ is an $\alpha c^{i+1}$-net of $V_i$ for each $i\in[0, k-1]$, it follows that $d(u, u_0) \leq \alpha c^k \sum_{i=0}^{k-1} \frac{1}{c^i} < \alpha c^k \frac{c}{c-1}$. Similar bound holds for $d(v, v_0)$. Thus: 
	\begin{align*}
	d(u,v) &\leq d(u, u_0) + d(v, v_0) + d(u_0, v_0) 
	\leq \alpha c^k (\frac{2c}{c-1} + 1) = \alpha c^k (\frac{2}{c-1} + 3) 
	\leq \alpha c^{k+t}.
	\end{align*} 
	Hence ${u, v}$ is an edge in $\Rips^{\alpha c^{k+t}}(V_0)$.
	 
	Next, observe that the vertex map $\hv_{k+t}$ restricted on the set of vertices $V_k$ is exactly the same as the vertex map $\vmap_{k, t}:= \vmap_{k+t-1} \circ \cdots \circ \vmap_k$ (this vertex map induces the simplicial map $\hmap_{k,t}$ in the diagram). Namely, for a vertex $u \in V_k \subseteq V_0$, $\hmap_{k,t}(u) = \hh_{k+t}(u)$. Thus $\hmap_{k,t} = \hh_{k+t} \circ j_k$. Hence the bottom triangle commutes both at the complex and the homology level. 
	
	We now consider the top triangle. Specifically, we prove that the map $j_{k} \circ \hh_{k}$ is contiguous to the inclusion map $i_k$.
	Since two contiguous maps induce the same homomorphisms at the homology level, the top triangle commutes at the homology level. 
	
	Indeed, given a simplex $\sigma \in \Rips^{\alpha c^k}(V_0)$, we need to show that vertices from $i_k(\sigma) \cup j_k\circ\hh_{k}(\sigma)$ span a simplex in $\Rips^{\alpha c^{k+t}}(V_0)$. 
	Since both are Rips complexes and $i_k$ and $j_k$ are inclusion maps, we only need to prove that for any two vertices $u$ and $v$ from $\sigma \cup \hh_{k}(\sigma)$, $d(u, v) \leq \alpha c^{k+t}$ (namely, ($u$,$v$) is an edge in $\Rips^{\alpha c^{k+t}}(V_0)$). 
	If $u$ and $v$ are both from $\sigma$ or both from $\hh_{k}(\sigma)$, 
	then $d(u,v) \leq \alpha c^{k+t}$ trivially. 
	Otherwise, assume without loss of generality that $v \in \sigma$ and $u \in \hh_{k}(\sigma)$, where $u=\hv_{k}(u')$ for some $u' \in \sigma$. Since $V_{i+1}$ is an is an $\alpha c^{i+1}$-net of $V_i$ for each $i\in[0, k-1]$, it follows that $d(u, u') \leq \alpha c^k \sum_{i=0}^{k-1} \frac{1}{c^i} < \alpha c^k \frac{c}{c-1}$. 
	One then has 
	\begin{eqnarray*}
		d(u,v) &\leq& d(u,u') + d(u',v) 
		\leq  
		\alpha c^k \frac{c}{c-1} + \alpha c^k = \alpha c^k \frac{2c-1}{c-1} 
		<  \alpha c^k (\frac{2}{c-1}+3) \leq \alpha c^{k+t}.
	\end{eqnarray*} 
	Thus $i_k$ is contiguous to $j_k \circ \hh_k$ and the lemma follows. 
\end{proof}

The above result implies that the persistence modules induced by 
sequences (\ref{SimBafiltration}) and (\ref{Ripsfiltration}) are weakly $\log c^t$-interleaved at the log-scale. Since $t \ge \log_c(\frac{2}{c-1}+3)$, we have $c^t \ge \frac{2}{c-1} + 3$.
By Theorem 4.3 of \cite{CSGGO2009}, we conclude with the following: 
\begin{theorem}
	The persistence diagram of the sequence (\ref{SimBafiltration}) provides a $3 \log(\frac{2}{c-1} + 3) $-approximation of the persistence diagram 
	of the sequence (\ref{Ripsfiltration}) at the log-scale for $c > 1$.
\end{theorem}

\section{Experiments}
\label{sec:Exp}

In this section, we report some experimental results of SimBa on large high dimensional data sets from other fields such as image processing, machine learning, and computational biology. For most of the data sets, previous approaches are not efficient enough to finish processing. They either ran out of memory (`$\infty$' in size) or ran more than one day (`$\infty$' in time). Table \ref{fig:table} at the end of this section provides the cumulative size and time cost for all four approaches mentioned in this paper. All approaches are implemented in C++. Note that we only compute persistences up to dimension 2 (which means we build simplicial complexes up to dimension 3). For Sparse Rips with GUDHI and Sparse Rips with Simpers, we choose parameter $\eps = 0.8$ which gives the best performance while not sacrificing  much of the approximation quality. For Batch-collapsed Rips with Simpers, we choose $c = 1.5$ which appears to reach a good trade-off between efficiency and quality. For SimBa, we choose $c = 1.1$ which in practice appears to have best quality -- note that the choice of $c$ does not seem to change the empirical efficiency much as Figure \ref{fig:FE_cd_compare}
illustrates. All experiments were run on a $64$-bit Windows machine with a 3.50GHz Intel processor and 16GB RAM. 

\paragraph*{Data with ground truth} 
We first test with two data sets whose ground truth persistences are known. They help demonstrate that SimBa works properly and efficiently in practice. All persistence barcodes shown in Figure \ref{fig:data_w_GT_pers} are original and not cleaned up.

We first consider a uniform sample of $22500$ points from
a Klein bottle in $\mathbb{R}^4$, and use SimBa to compute its barcode 
which is shown in Figure \ref{fig:KB_d_pers}. There are two main bars for $H_1$ and one for $H_2$ as expected. 

\begin{figure}[t]
	\centering
	\begin{subfigure}{0.32\textwidth}
		\centering\includegraphics[height=0.8\textwidth]{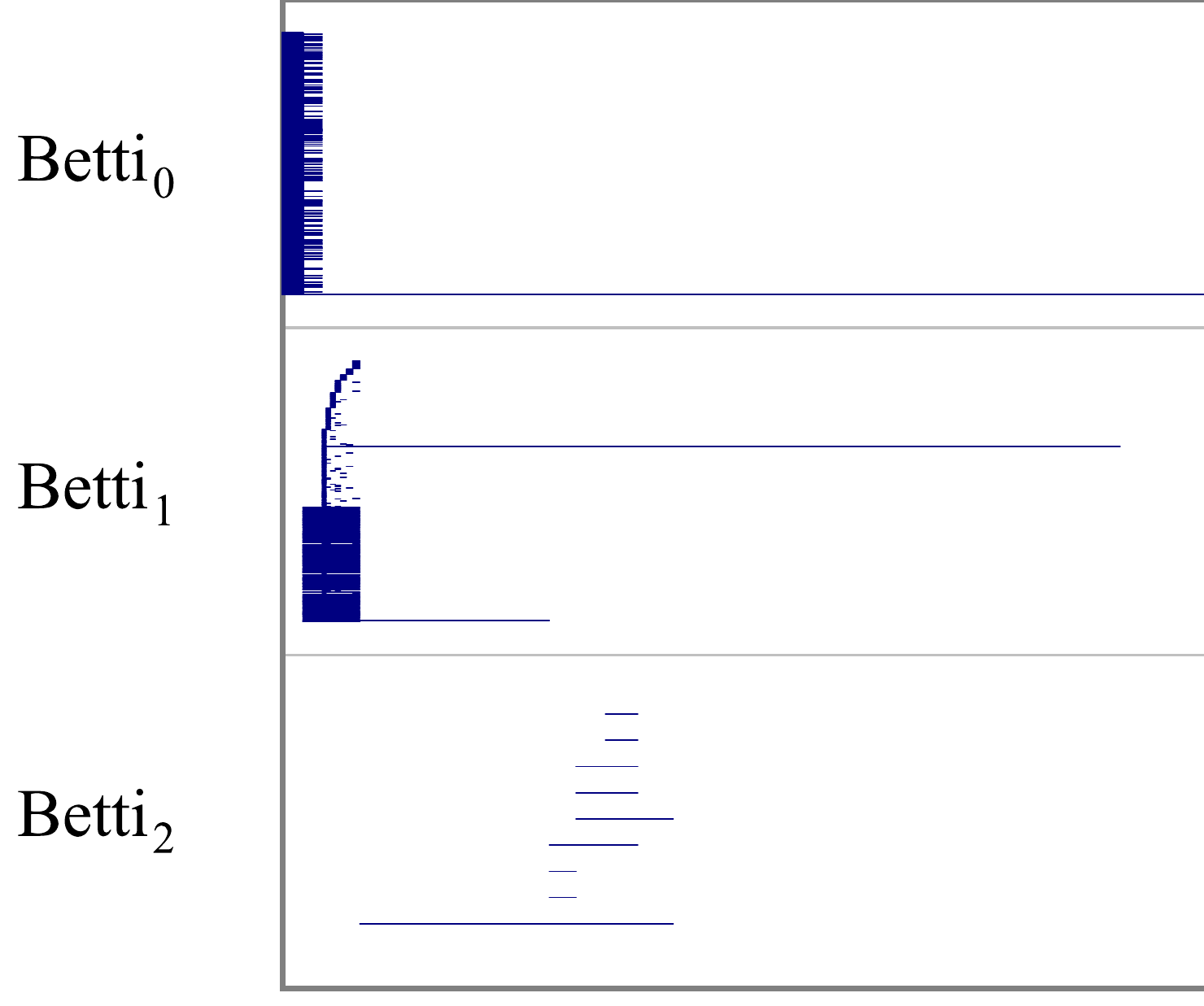}
		\caption{Klein Bottle in $\mathbb{R}^{4}$}\label{fig:KB_d_pers}
	\end{subfigure}
	\begin{subfigure}{0.32\textwidth}
		\centering\includegraphics[height=0.8\textwidth]{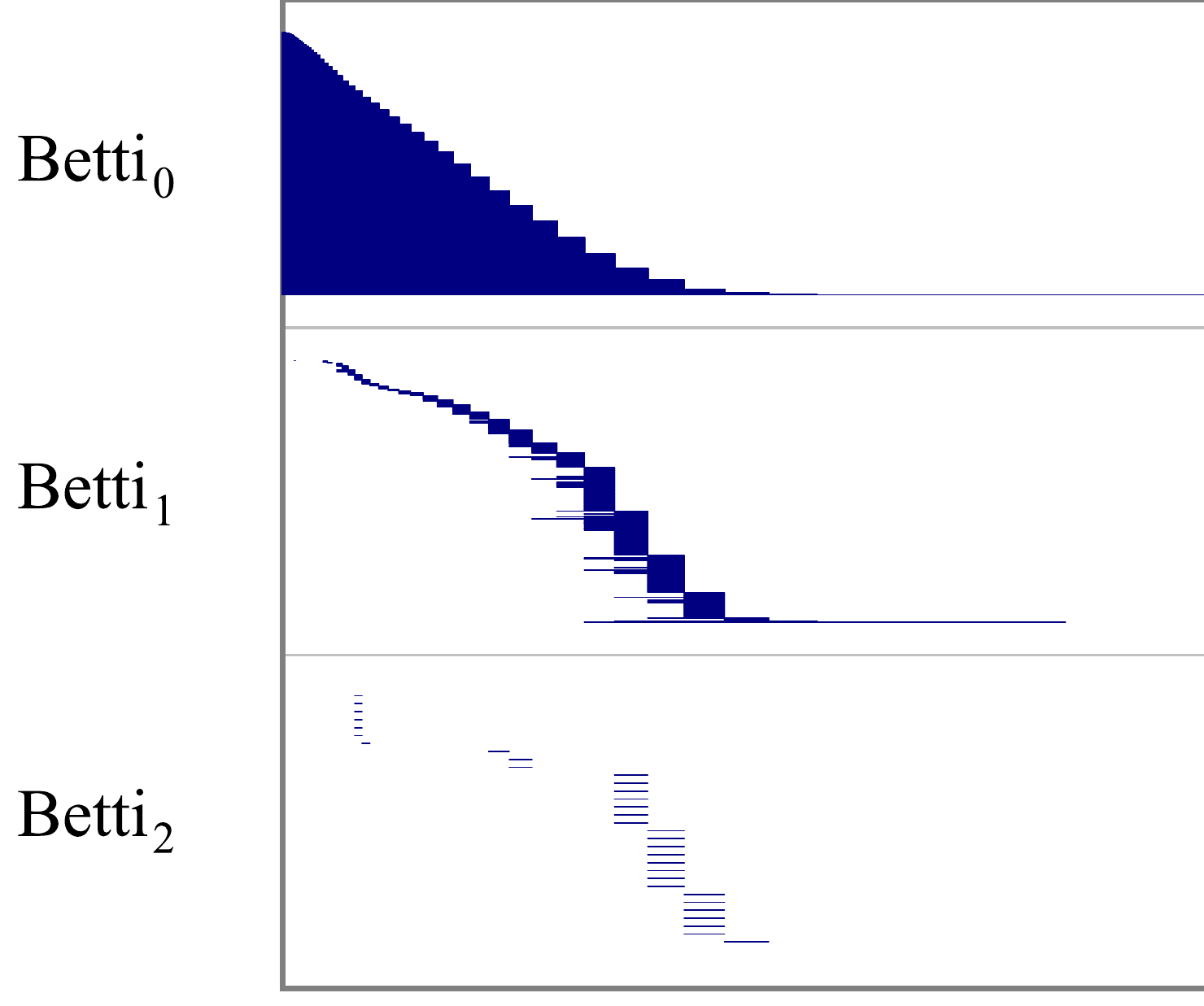}
		\caption{Primary Circle in $\mathbb{R}^{25}$}\label{fig:PC5_d_pers}
	\end{subfigure}
	\begin{subfigure}{0.32\textwidth}
		\centering\includegraphics[height=0.8\textwidth]{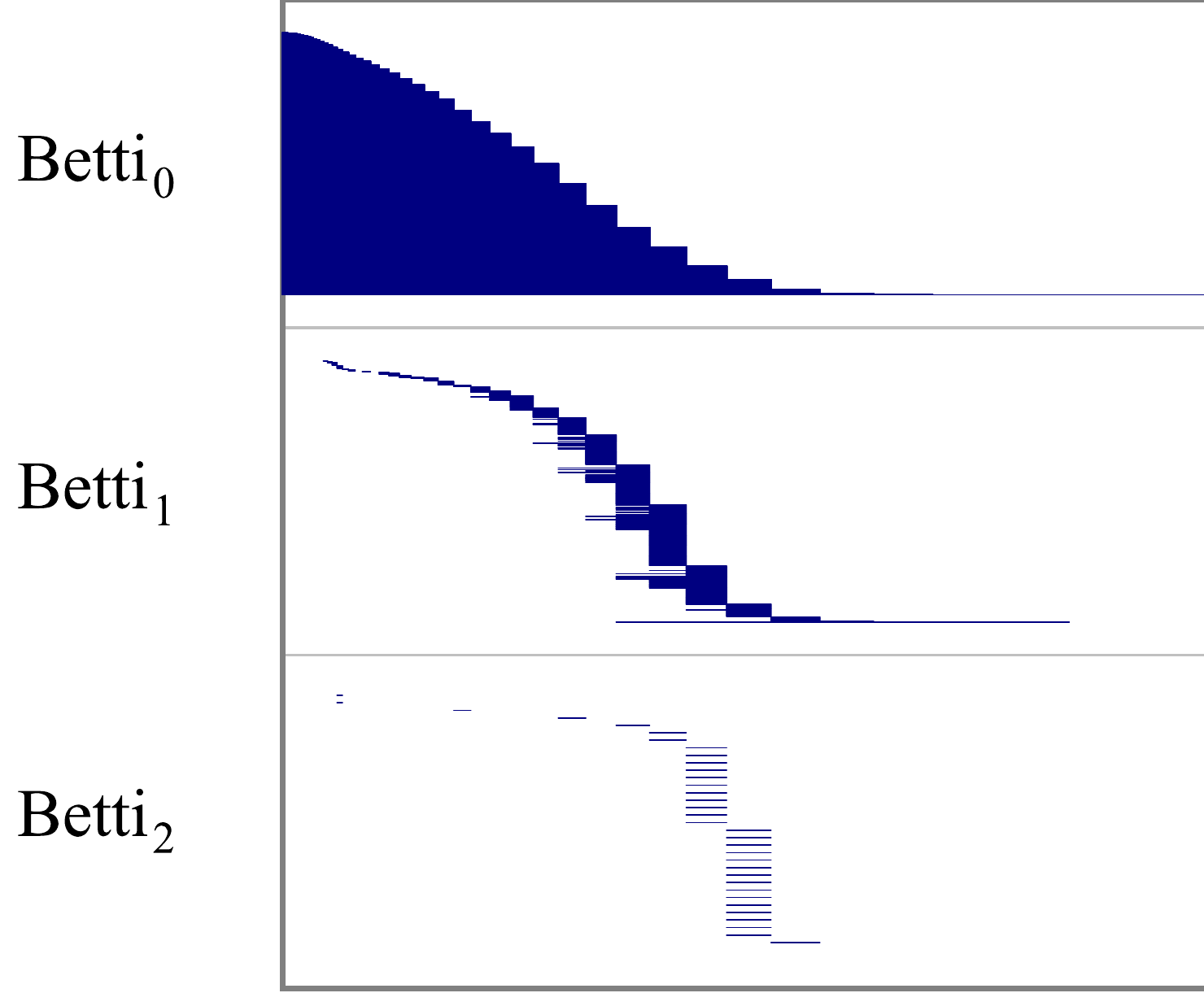}
		\caption{Primary Circle in $\mathbb{R}^{49}$}\label{fig:PC7_d_pers}
	\end{subfigure}

	\caption{Original persistence barcodes computed by SimBa on data sets with ground truth}
	\label{fig:data_w_GT_pers}
\end{figure}
Next, we consider the primary circle of natural image data in \cite{AC2009}, which has 15000 points. Each point is a 5$\times$5 or 7$\times$7 image patch, thus considered as a point in $\mathbb{R}^{25}$ or $\mathbb{R}^{49}$. From Figures \ref{fig:PC5_d_pers} and \ref{fig:PC7_d_pers}, we can see the primary circle bar for $H_1$ for data both in $\mathbb{R}^{25}$ and $\mathbb{R}^{49}$. All short bars for $H_2$ persist for only one batch step and thus 
can be regarded as noise. 

\paragraph*{Data without ground truth}
Next, we provide some experiments on the data sets whose ground truth persistences are not known. We used SimBa to compute their persistences and found some relatively long bars which are likely to be features and may worth further investigation by domain experts. The persistence barcodes shown in Figure \ref{fig:data_wo_GT_pers} and \ref{fig:GP_d_pers_den} are denoised for $H_1$. The rest of Figure \ref{fig:GP_compare_pers} are original.

We first take the Gesture Phase Segmentation data set \cite{MWP2014} from UCI machine learning repository \cite{L2013}. This data set was used in \cite{MPL2016}. It comprises of features extracted from 7 videos with people gesticulating. Each video is represented by a raw file
that contains the positions of hands, wrists, head, and spine of the user in each frame. We took the raw file from video A1 of 1747 frames. Since there are six sensors each with x, y, z coordinates, we have in total 1747 points in $\mathbb{R}^{18}$. There are five gesture phases in the videos: rest, preparation, stroke, hold, and retraction. Indeed, there are five long bars for $H_0$ in \ref{fig:GP_pers} (although they overlap and do not stand out in the picture), which seems to match the five clusters of different phases. We see some 
long bars for $H_1$, which could be created due to periodic patterns in these gesture movements. 

Another data set is the Survivin protein data from \cite{HPRPBLW2014}. There are totally 252996 protein conformations and each conformation is considered as a point in $\mathbb{R}^{150}$. We used PCA to reduce the data dimension to 3. We ran SimBa on both data sets and show the barcodes in Figure \ref{fig:SV_150_pers} and \ref{fig:SV_3_pers}. We can see that there are some long bars for $H_1$.

\begin{figure}
	\centering
	\begin{subfigure}{0.32\textwidth}
		\centering\includegraphics[height=0.8\textwidth]{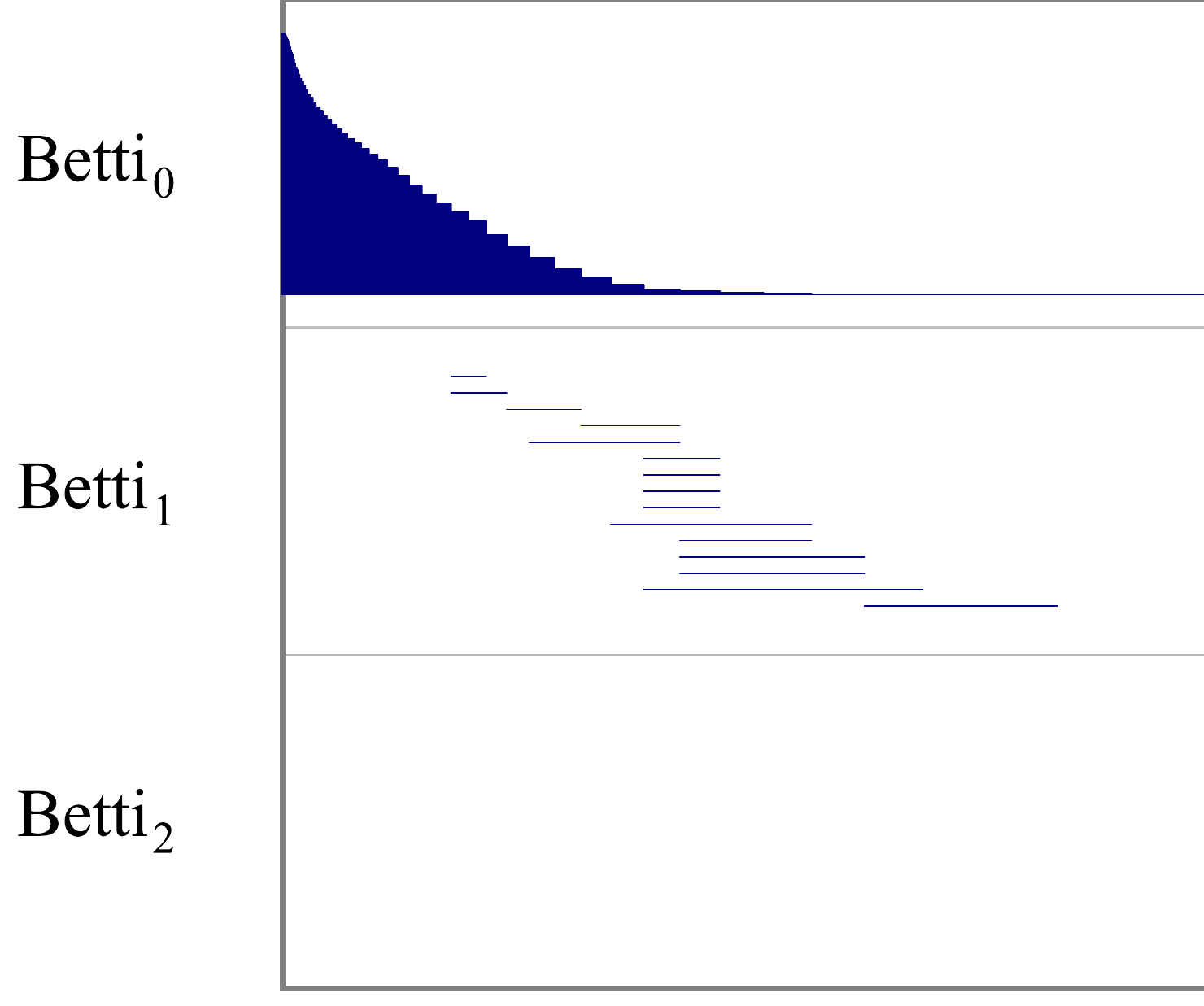}
		\caption{Gesture Phase data in $\mathbb{R}^{18}$}\label{fig:GP_pers}
	\end{subfigure}
	\begin{subfigure}{0.32\textwidth}
		\centering\includegraphics[height=0.8\textwidth]{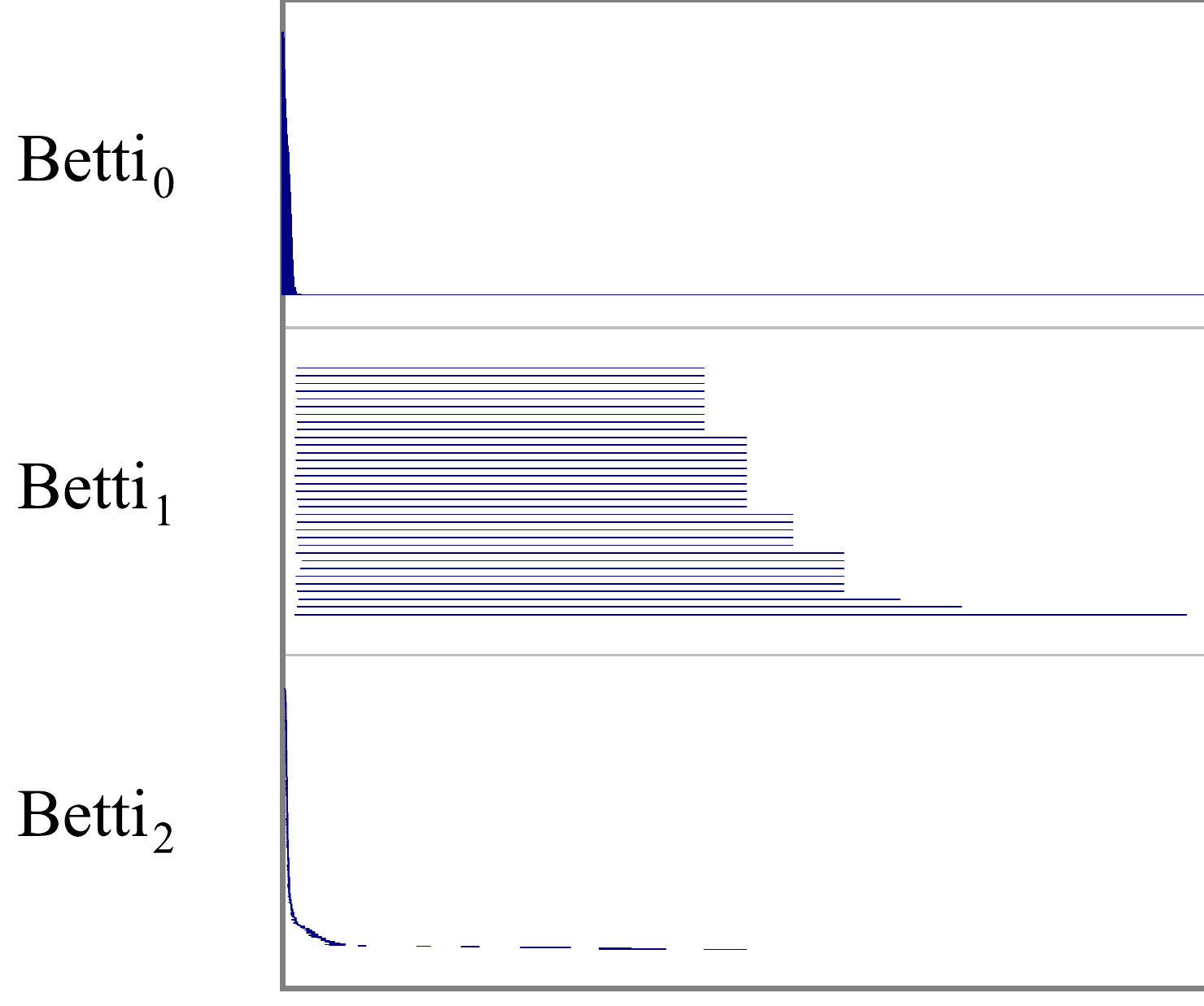}
		\caption{Survivin data in $\mathbb{R}^{3}$}\label{fig:SV_3_pers}
	\end{subfigure}
	\begin{subfigure}{0.32\textwidth}
		\centering\includegraphics[height=0.8\textwidth]{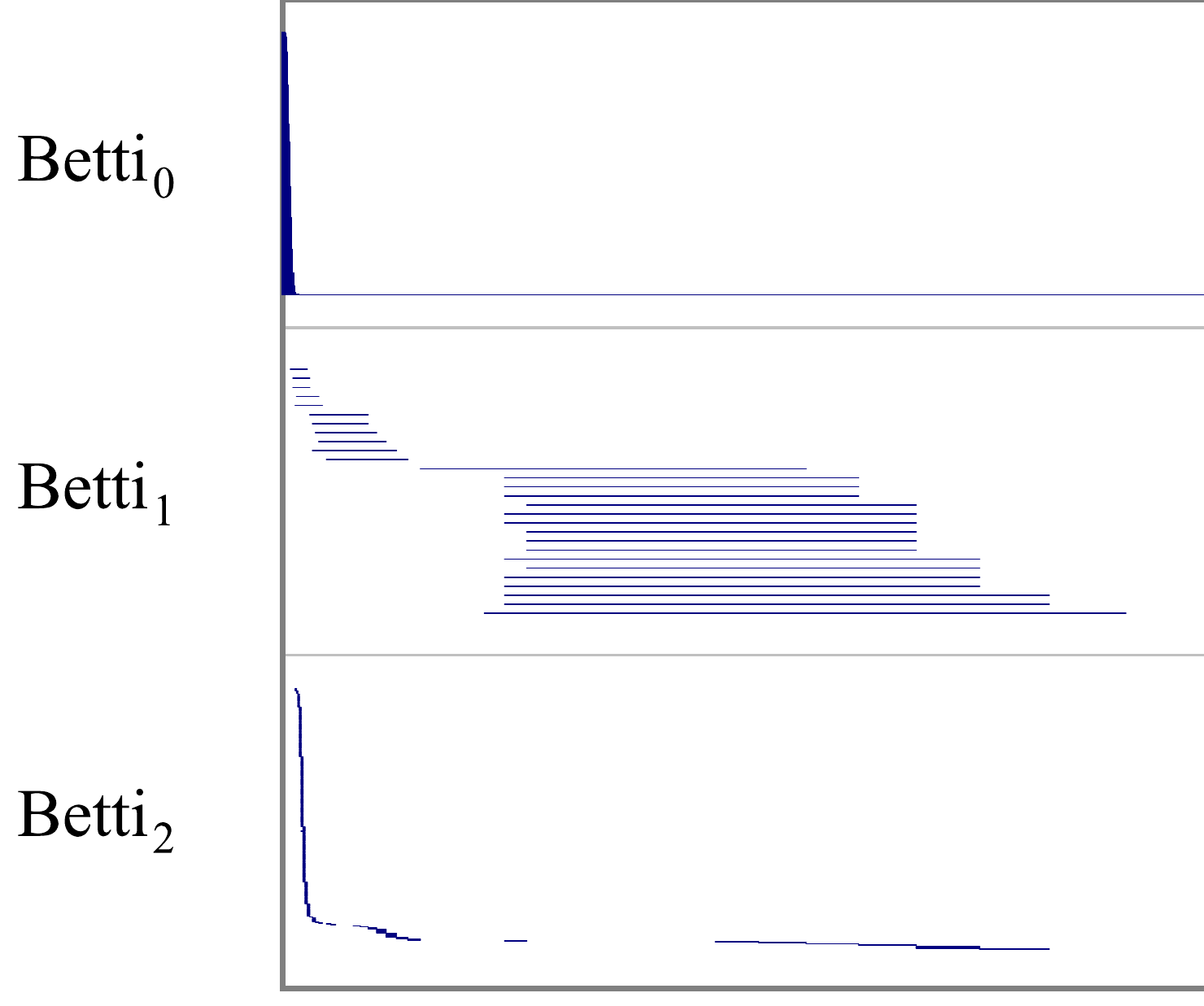}
		\caption{Survivin data in $\mathbb{R}^{150}$}\label{fig:SV_150_pers}
	\end{subfigure}
	\caption{Denoised persistence barcodes computed by SimBa on data sets without ground truth}
	\label{fig:data_wo_GT_pers}
\end{figure}
\begin{figure}
	\centering
	\begin{subfigure}{0.24\textwidth}
		\centering\includegraphics[height=0.8\textwidth]{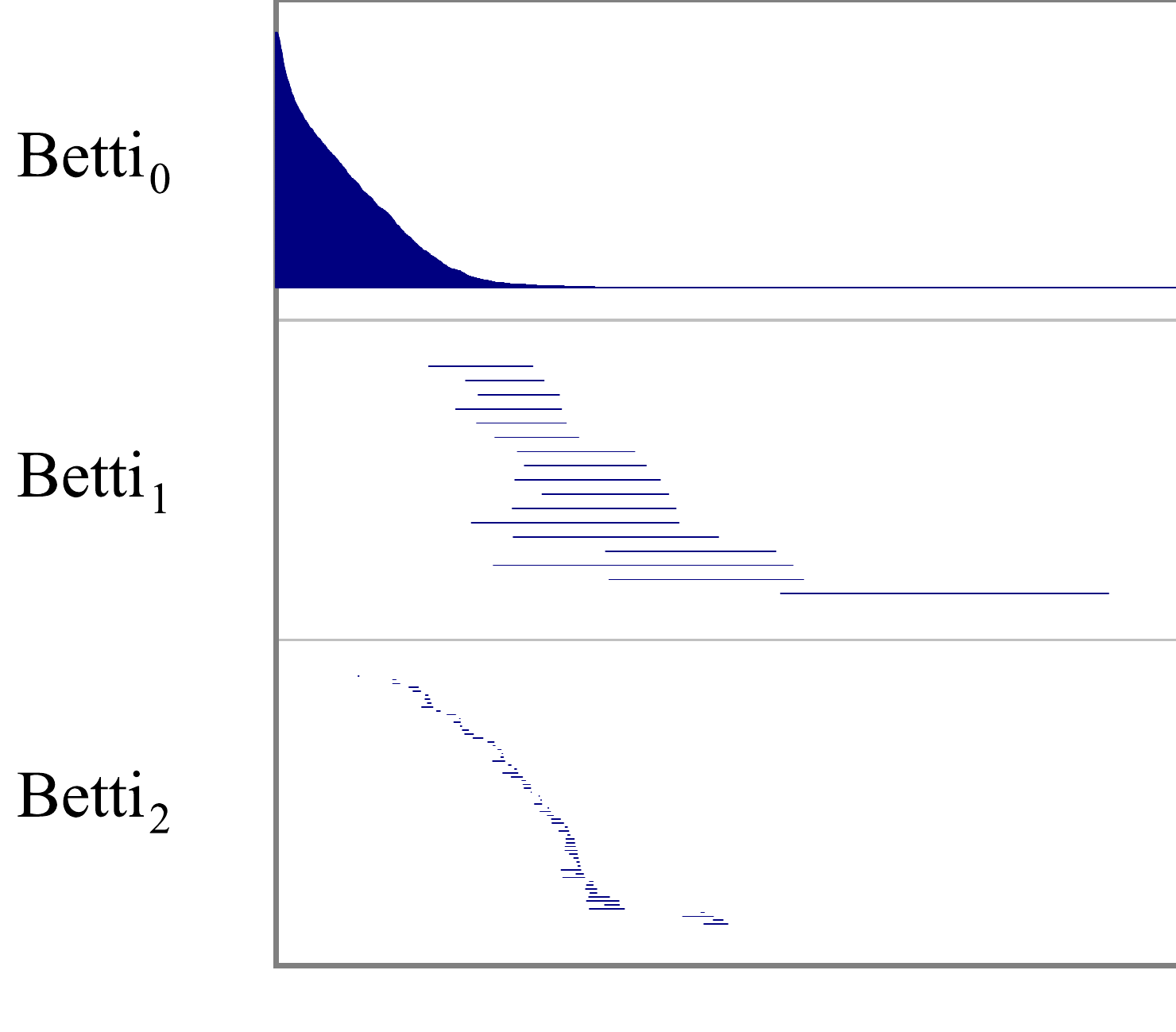}
		\caption{S.R.+GUDHI}\label{fig:GP_a_pers}
	\end{subfigure}
	\begin{subfigure}{0.24\textwidth}
		\centering\includegraphics[height=0.8\textwidth]{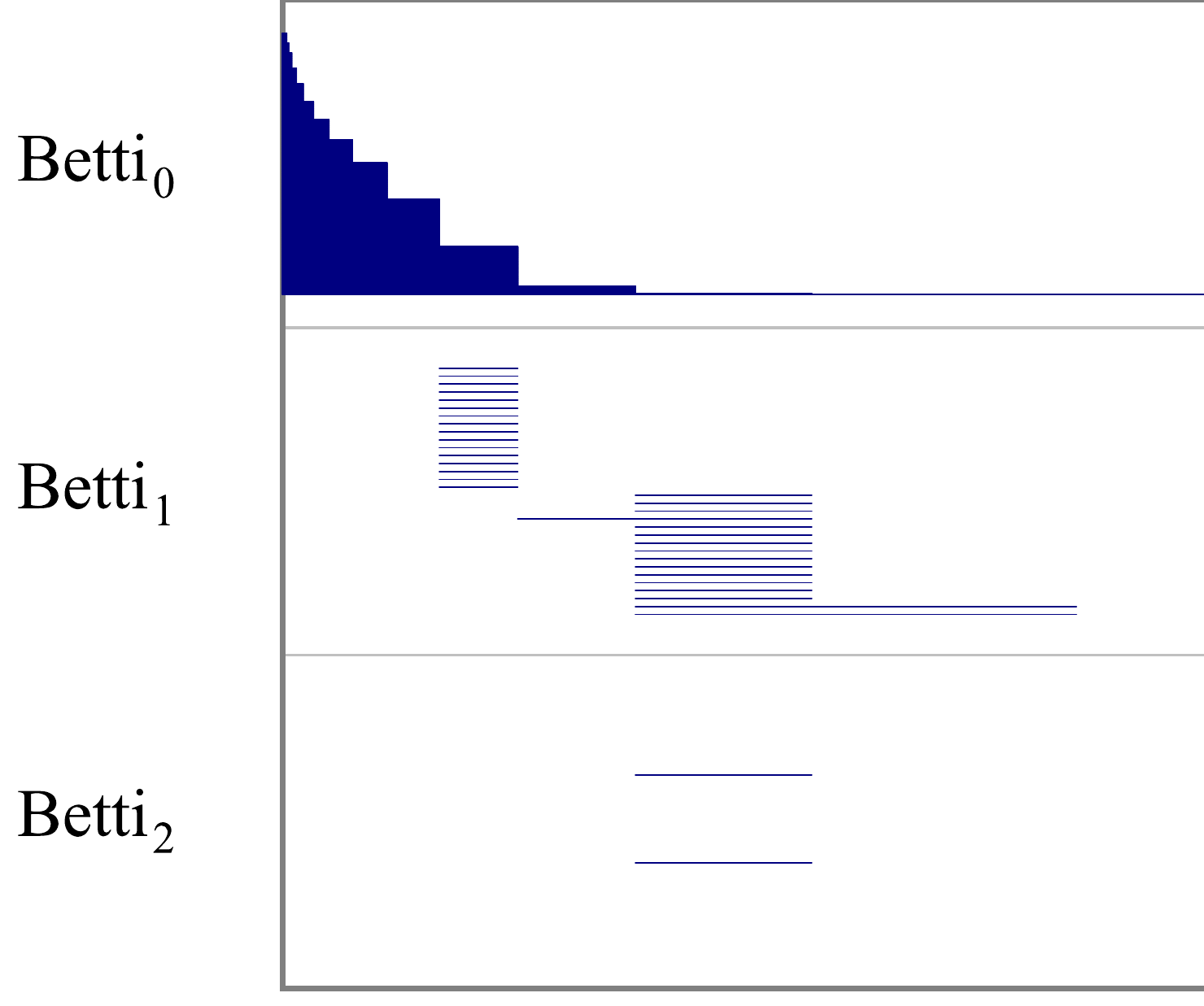}
		\caption{B.R.+Simpers}\label{fig:GP_c_pers}
	\end{subfigure}
	\begin{subfigure}{0.24\textwidth}
		\centering\includegraphics[height=0.8\textwidth]{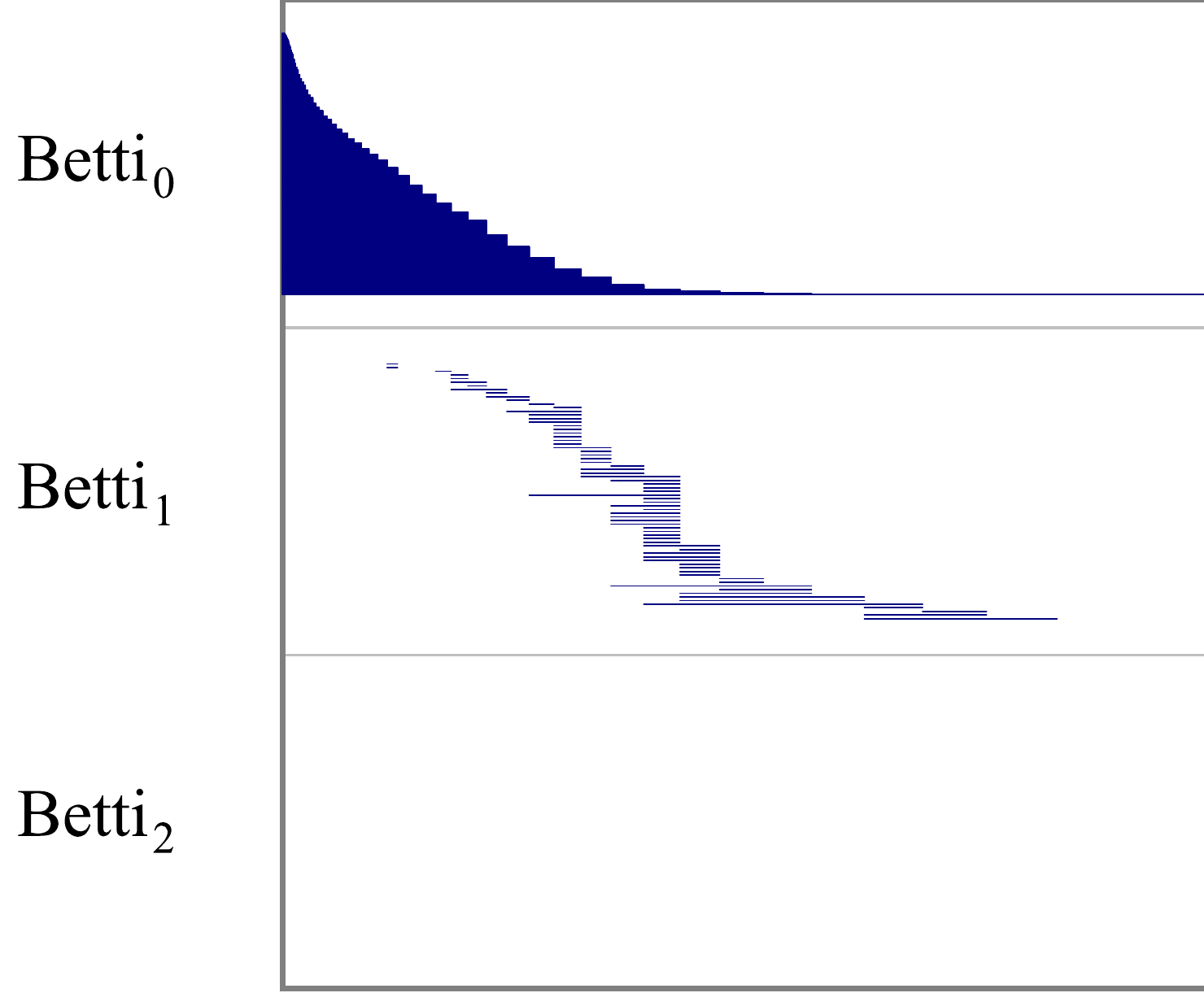}
		\caption{SimBa}\label{fig:GP_d_pers}
	\end{subfigure}
	\begin{subfigure}{0.24\textwidth}
		\centering\includegraphics[height=0.8\textwidth]{GP_d_pers_den.pdf}
		\caption{SimBa (denoised)}\label{fig:GP_d_pers_den}
	\end{subfigure}
	\caption{Persistence barcodes computed by different approaches on Gesture Phase Segmentation data in $\mathbb{R}^{18}$}
	\label{fig:GP_compare_pers}
\end{figure}
\paragraph*{Performance results}
We provide the performance results for all data sets
mentioned in Table~\ref{fig:table}, which includes cumulative size and time cost of each approach. The time is obtained
by adding the time to construct the complexes and the time to compute persistence. {\bf  S.R.+GUDHI}, {\bf S.R.+Simpers}, {\bf B.R.+Simpers} and {\bf SimBa} stand for Sparse Rips plus GUDHI, Sparse Rips plus Simpers, Batch-collapsed Rips plus Simpers, and SimBa respectively. {\bf Mother}, {\bf KlBt}, {\bf PrCi25}, {\bf PrCi49}, {\bf GePh}, {\bf Sur3} and {\bf Sur150} stand for MotherChild model, Klein Bottle, Primary Circle in $\mathbb{R}^{25}$, Primary Circle in $\mathbb{R}^{25}$, Gesture Phase Segmentation data, Survivin protein data in $\mathbb{R}^{3}$ and in $\mathbb{R}^{150}$ respectively. Each data set has size $n$, ambient dimension $D$, and intrinsic dimension $d$. The
symbol $\infty$ means that the program either ran out of memory or 
did not finish after a day. From the table, we can see that SimBa out-performed the other three approaches significantly. Notice that for those larger cases of SimBa, the nearest neighbor search operations (ANN) usually take most of time and become the bottleneck. This is why {\bf Sur150} costs much more time than {\bf PrCi49} while its cumulative size is smaller. It would be an interesting
future work to make nearest neighbor search more efficient so that SimBa performs better even for such cases.
\begin{table}[h]
	\centering
	\setlength{\tabcolsep}{8pt}
	\begin{tabular}{cccc}
		\hspace{3.8cm} {\bf S.R.+GUDHI} & {\bf S.R.+Simpers} & {\bf B.R.+Simpers} & \hspace{0.3cm} {\bf SimBa} \\
	\end{tabular}
	\setlength{\tabcolsep}{3.5pt}
	\begin{tabular}{| l c c c | c c | c c | c c | c c |}
		\hline
		
		Data & $n$ & $D$ & $d$ & size & time(s) & size & time(s) & size & time(s) & size & time(s) \\
		\hline
		
		{\bf Mother} & $23075$ & $3$ & $2$ & $43.5\cdot 10^6$ & $350$
		& $43.5\cdot 10^6$ & $463.7$ & $2.3\cdot 10^6$ & $42.3$ & $104701$ & $8.8$ \\
		
		{\bf KlBt} & $22500$ & $4$ & $2$ & $20.9\cdot 10^6$ & $205.3$
		& $20.9\cdot 10^6$ & $303.5$ & $440049$ & $8$ & $78064$ & $6.6$ \\
		
		{\bf PrCi25} & $15000$ & $25$ & ? & $\infty$ & $-$ & $\infty$ & $-$ & $-$ & $\infty$ & $4.8\cdot 10^6$ & $216$\\
		
		{\bf PrCi49} & $15000$ & $49$ & ? & $\infty$ & $-$ & $\infty$ & $-$ & $-$ & $\infty$ & $10.2\cdot 10^6$ & $585$\\
		
		{\bf GePh} & $1747$ & $18$ & ? & $45.6\cdot 10^6$ & $282.5$ & $45.6\cdot 10^6$ & $432.8$ & $1.4\cdot 10^6$ & $29$ & $7145$ & $0.83$\\
		
		{\bf Sur3} & $252996$ & $3$ & ? & $\infty$ & $-$ & $\infty$ & $-$ & $15.7\cdot 10^6$ & $1056.4$ & $915110$ & $1079.6$\\
		
		{\bf Sur150} & $252996$ & $150$ & ? & $\infty$ & $-$ & $\infty$ & $-$ & $-$ & $\infty$ & $3.1\cdot 10^6$ & $5089.7$\\
		
		\hline
	\end{tabular}
	\caption{cumulative size and time cost}
	\label{fig:table}
\end{table}

\paragraph*{Acknowledgment.} This work has been supported by
NSF grants CCF-1318595 and CCF-1526513.
\newpage
\bibliographystyle{plain}
\bibliography{ref}

\end{document}